\documentclass[12pt]{article}

\RequirePackage{amsmath, amsthm, amsfonts, graphicx, bm, natbib}
\RequirePackage{setspace, microtype, dsfont, subcaption, appendix}
\RequirePackage{multirow, colortbl}
\RequirePackage{hyperref}
\RequirePackage[raggedright]{titlesec}
\RequirePackage{tikz}
\RequirePackage{enumitem}
\usetikzlibrary{arrows}

\DisableLigatures[f]{encoding=*, family=*}

\numberwithin{equation}{section}
\numberwithin{figure}{section}
\numberwithin{table}{section}

\providecommand{\keywords}[1]{\textbf{Keywords:} #1}

\newtheorem{lemma}{Lemma}

\def\by{\bm{y}}
\def\bx{\bm{x}}
\def\bz{\bm{z}}
\def\bB{\bm{B}}
\def\bu{\bm{u}}
\def\bR{\bm{R}_{\beps}}
\def\Leps{\bm{\mathcal{L}}_{\beps}}

\def\eps{\epsilon}
\def\d{\delta}

\def\m{\mu}
\def\k{\kappa}
\def\s{\sigma}
\def\l{\lambda}
\def\a{\alpha}
\def\b{\beta}
\def\vareps{\varepsilon}
\def\ba{\boldsymbol \alpha}
\def\beps{\boldsymbol \epsilon}
\def\bmu{\boldsymbol \mu}
\def\bt{\boldsymbol \theta}
\def\bb{\boldsymbol \beta}
\def\Sig{\boldsymbol \Sigma}
\def\bL{\boldsymbol \Lambda}
\def\vt{\boldsymbol \vartheta}
\def\bSigeps{\Sig_{\beps}}
\def\N{\mathcal{N}}
\def\diag{\textrm{diag}}
\def\vechL{\textrm{vechL}}
\def\E{\mathbb{E}}
\def\P{\mathbb{P}}

\title{\bf Efficient data augmentation for multivariate probit models with panel data: An application to general practitioner decision-making about contraceptives}

\author{Vincent Chin\footnote{Communicating author: \href{mailto:vincent.chin@student.unsw.edu.au}{\tt vincent.chin@student.unsw.edu.au}} \thanks{Australian Research Council Centre of Excellence for Mathematical \& Statistical Frontiers, The University of Melbourne, Victoria 3010, Australia.} \thanks{School of Mathematics and Statistics, University of New South Wales, Sydney 2052, Australia.} ,
David Gunawan\footnotemark[2] \thanks{School of Economics, Australian School of Business, University of New South Wales, Sydney 2052, Australia.} ,
Denzil G. Fiebig\footnotemark[2] \footnotemark[4] , \\
Robert Kohn\footnotemark[2] \footnotemark[4] { and}
Scott A. Sisson\footnotemark[2] \footnotemark[3]}

\date{}

\begin{document}
\maketitle
\begin{abstract}
\onehalfspacing
This article considers the problem of estimating a multivariate probit model in a panel data setting with emphasis on sampling a high-dimensional correlation matrix and improving the overall efficiency of the data augmentation approach.  We reparameterise the correlation matrix in a principled way and then carry out efficient Bayesian inference using Hamiltonian Monte Carlo.  We also propose a novel antithetic variable method to generate samples from the posterior distribution of the random effects and regression coefficients, resulting in significant gains in efficiency.  We apply the methodology by analysing stated preference data obtained from Australian general practitioners evaluating alternative contraceptive products.  Our analysis suggests that the joint probability of discussing combinations of contraceptive products with a patient shows medical practice variation among the general practitioners, which indicates some resistance to even discuss these products, let alone recommend them.
\end{abstract}
\keywords{Antithetic variable; Bayesian inference; Correlated binary data; Hamiltonian Monte Carlo; Panel data.}

\doublespacing
\newpage
\section{Introduction} \label{sec:introduction}

Bayesian inference for the multivariate probit (MVP) model is usually performed using the data augmentation representation of \cite{chib1998analysis}, whereby the latent variables indicating the observed outcomes are normally distributed.  For unique identification of the regression parameters, the covariance matrix of these latent normal random variates is assumed to be a correlation matrix $\bR$.  However, Monte Carlo sampling for $\bR$ in a Bayesian context is difficult due to the restrictions on the diagonal entries and the requirement that the matrix $\bR$ must be positive definite.

This article presents three contributions, two  methodological and the third a subject matter one.  The first methodological contribution provides an improved method for sampling the potentially high dimensional correlation matrix $\bR$ within a Markov chain Monte Carlo (MCMC) algorithm.  In order to circumvent the positive definiteness restriction imposed on a correlation matrix, we adopt the reparameterisation strategy of \cite{smith2013bayesian} which re-expresses $\bR$ as an unconstrained Cholesky factor $\Leps$.  This maps the manifold space of a correlation matrix to a Euclidean space, which  improves posterior simulation while keeping the number of unknown parameters the same.  A prior distribution is then specified on $\Leps$ such that the implied marginal densities of the correlation coefficients are uniform on $(-1,1)$.  We employ the Hamiltonian Monte Carlo (HMC) algorithm \citep{neal2011mcmc} to sample the high dimensional $\Leps$ efficiently, thereby avoiding the slow exploration of parameter space by random walk updates as in \cite{smith2013bayesian}.

The second methodological contribution is to introduce antithetic sampling, based on the work of \cite{hammersley1956new}, into the Metropolis-Hastings (MH) literature.  In order to implement this idea, we specify the proposal distribution of parameter update as a deterministic function.  Here, the generated samples will be super-efficient in terms of the reduction in variance of the Monte Carlo estimates compared to the same estimates constructed from uncorrelated samples.  Although the chain update proposal is deterministic, the convergence properties are not compromised when this is embedded within a larger system of MCMC sampling.  Our proposed methodology is motivated by the over-relaxation algorithm \citep{adler1981over, barone1990improving}, and is similar to the idea built within the framework of HMC in \cite{pakman2014exact}.  However, our proposed sampler is different from these methods in two main aspects.  First, there is no randomness in the proposal distribution for parameter updates in our method, whereas theirs still retain a certain degree of stochasticity.  Second, we introduce perfect negative correlation between successive MCMC samples via the deterministic proposal, while they suggest partial or zero dependence between the samples.  Results based on our real data application document a significant improvement of up to a 16 times performance gain in the mixing behaviour of the Markov chain, thereby lowering the autocorrelation between the iterates.  The computing time of the algorithm is also marginally reduced due to the deterministic sampling.

Our methodological development is motivated by the staged stated preference panel data collection described in \cite{fiebig2017consideration}, which is used to study the decision-making of Australian general practitioners (GPs) about female contraceptive products.  Here, the authors used the data from the third and final stage, whereas we explore outcomes from the second stage.  This second stage relates to the question of which particular contraceptive products GPs would discuss with a female patient, defined by a vignette that is part of the experimental design.  Separate univariate analyses on each product would ignore possible complex dependence structures that are useful in exploring which particular bundles of products are discussed with patients.  This is important here because in any correlated choice problem there may be multiple close substitutes, which makes joint rather than marginal probabilities more relevant.  Therefore, we model the GPs' choices by an MVP model.  Inspection of the resulting graphical model describing this interaction between products lends support to the suitability of a multivariate approach.  By using the MVP model, we are able to compute the joint probability of specific product bundles being discussed with a patient.  Posterior estimation of this probability, based on a patient with certain socio-economic and clinical characteristics, reveals differing views among the GPs in the sample on the suitability of long acting contraceptive choices.  This variability is known as medical practice variation in the health industry, whereby the decision making of GPs is influenced by both their personal characteristics such as gender, age and qualifications, as well as other unobservables that we model as random effects.

The rest of the paper is organised as follows. Section~\ref{sec:multivariate} describes the MVP model with random effects and reviews previous research associated with sampling $\bR$.  Section~\ref{sec:efficient} presents our proposed methodology of sampling $\bR$, and Section~\ref{sec:deterministic} outlines the antithetic sampling technique whose efficiency is illustrated via simulation studies in Section~\ref{sec:simulation}.  Section~\ref{sec:application} provides our analysis of the discussion preference data of contraceptive products by Australian GPs, and Section~\ref{sec:conclusion} concludes.  Appendices~\ref{app:sampling scheme}--\ref{app:covariance} provide further details on the contraceptive product data analysis.

\section{Multivariate probit model with random effects} \label{sec:multivariate}

The MVP model has been used extensively to model correlated binary data \citep{gibbons1998health, buchmueller2013preference}.  Let $\by_{it}=(y_{1,it}, \dotsc, y_{D,it})^\top$ be a vector of $D$ correlated binary outcomes for individual $i=1,\dotsc,P$ at time period $t$, for $t=1, \dotsc, T$.  The latent variable representation of the MVP model, using the data augmentation approach of \cite{albert1993bayesian}, is given by
\begin{gather}
\by^\ast_{it} = \ba_i + \bB \bx_{it} + \beps_{it},
\label{eqn:MVP} \\
\ba_i = (\a_{1,i}, \dotsc, \a_{D,i})^\top \stackrel{iid}{\sim} \N({\bf 0}, \Sig_{\ba}), \\
\beps_{it} = (\eps_{1,it}, \dotsc, \eps_{D,it})^\top \stackrel{iid}{\sim} \N({\bf 0}, \bR),
\end{gather}
for $i=1,\dotsc,P, t=1,\dotsc,T,$ where $\by^\ast_{it}=(y^\ast_{1,it}, \dotsc, y^\ast_{D,it})^\top$ is a continuous latent variable, $\ba_i$ is a $D$-vector of outcome-specific random effects for individual $i$ allowing for heterogeneity between individuals, $\bx_{it}=(1, x_{1,it}, \dotsc, x_{K-1,it})^\top$ is an exogenous variable, $\bB$ is a $D \times K$ matrix of regression coefficients and $\beps_{it}$ is a $D$-vector correlated error term which models the dependence structure between outcomes.  The variable $\bx_{it}$ is assumed to be uncorrelated with both $\ba_i$ and $\beps_{it}$.  This is entirely appropriate in the stated preference case that is our motivating analysis but relaxing the assumption of exogenous $\bx_{it}$ represents a useful extension.  In order for $\bB$ to be uniquely identified \citep{chib1998analysis}, $\bR$ is set to be a correlation matrix.  The observed outcome $\by_{it}$ is defined to be dependent on the latent variable $\by^\ast_{it}$ via the relationship
\begin{equation}
y_{d,it} = \mathds{1}(y^\ast_{d,it}>0), \quad d=1, \dotsc, D,
\label{eqn:latent}
\end{equation}
where $\mathds{1}(\bm E)$ is an indicator function which takes value 1 if the event $\bm E$ occurs and 0 otherwise.  Let $\by=\{\by_{it}; i=1,\dotsc,P, t=1,\dotsc,T\}$ be the set of observed discrete outcomes.  The density of the latent variables $\by^\ast$ conditional on the random effects $\ba_{1:P}=(\ba_1,\dotsc,\ba_P)$ is given by
\begin{equation}
p(\by^\ast|\ba_{1:P}, \bt) = \prod_{i=1}^P \prod_{t=1}^T \phi(\by^\ast_{it}; \bmu_{it}, \bR),
\label{eqn:likelihood}
\end{equation}
where $\bt:=(\bB, \bR, \Sig_{\ba})$ denotes the vector of model parameters, $\bmu_{it}=\ba_i + \bB \bx_{it}$ and $\phi$ is the multivariate normal density function.

Following the specification of the MVP model in (\ref{eqn:MVP})--(\ref{eqn:latent}), the posterior density is
\begin{equation}
\pi(\by^\ast, \ba_{1:P}, \bt|\by) = \frac{p(\by|\by^\ast, \ba_{1:P}, \bt)p(\by^\ast|\ba_{1:P},\bt)p(\ba_{1:P}|\bt)p(\bt)}{p(\by)},
\label{eqn:posterior}
\end{equation}
where $p(\by)$ is the marginal likelihood, $p(\bt)$ is the prior on the model parameters $\bt$ and
\begin{equation}
p(\by|\by^\ast, \ba_{1:P}, \bt) = \prod_{i=1}^P \prod_{t=1}^T \prod_{d=1}^D \bigg(\mathds{1}(y_{d,it}=0)\mathds{1}(y^\ast_{d,it} \leq 0) + \mathds{1}(y_{d,it}=1)\mathds{1}(y^\ast_{d,it} > 0)\bigg).
\label{eqn:augmented posterior}
\end{equation}
Useful conjugate priors are available for $\bB$ (or $\bb=\text{vec}(\bB)$) and $\Sig_{\ba}$ which simplifies MCMC sampling, but it is difficult to posit a suitable prior for $\bR$.

\subsection{Prior choice for the correlation matrix \texorpdfstring{$\bR$}{bReps}} \label{subsec:prior}

\cite{barnard2000modeling} decompose a covariance matrix $\bSigeps$ as $\bm{S}\bR\bm{S},$ where $\bm{S}$ is a diagonal matrix of standard deviations and $\bR$ is a correlation matrix.  They show that if $\bSigeps \sim \mathcal{IW}(\nu, \bm{I})$, i.e. an inverse-Wishart distribution with degrees of freedom $\nu$ and scale matrix $\bm{I}$, then the density of $\bR$ is
\begin{equation}
p(\bR) \propto |\bR|^{\frac{1}{2}(\nu-1)(D-1)-1}\Bigg(\prod_{i=1}^D|\bR(-i;-i)|\Bigg)^{-\frac{\nu}{2}},
\label{eqn:marginally uniform prior}
\end{equation}
where $\bR(-i;-i)$ denotes the $i$-th principal submatrix of $\bR$, that is $\bR$ with its $i$-th row and column removed.  We follow \cite{barnard2000modeling} and take (\ref{eqn:marginally uniform prior}) as the prior for $\bR$, which induces a modified Beta distribution on each off-diagonal element $r_{ij}$ of $\bR, i \neq j$.  In particular, the marginal densities of the $r_{ij}$ are uniform on $(-1,1)$ when $\nu=D+1$, which means that posterior inference is invariant to the ordering of the binary outcomes $\by$.  Furthermore, recent results in \cite{wang2018on} establish that for such a choice of $\nu$, the corresponding matrix of partial correlations $\rho_{kl}$ has the LKJ distribution of \cite{lewandowski2009generating} with unit shape parameter.  This means that all $\rho_{kl}$ are marginally distributed according to a $\textrm{Beta}(\frac{D}{2}, \frac{D}{2})$ distribution over $(-1,1)$ with both shape parameters $\frac{D}{2}$, which is informative in high dimensions because the Beta density increasingly concentrates around zero.  The informativity of $\rho_{kl}$ is useful in practical applications, where more often than not a sparse structure on the partial correlation matrix is desirable to suggest conditional independence.

The dependence structures imposed by the marginally uniform prior are less studied in the literature.  Since analytical results for these properties are limited \citep{tokuda2011visualizing}, we briefly illustrate these graphically instead.  The results obtained are based on correlation matrices of dimension $D=4$ but they can be generalised to higher dimensions.  We generate $10^7$ samples from (\ref{eqn:marginally uniform prior}) with $\nu=D+1$ by normalising the covariance matrices drawn from an $\mathcal{IW}(D+1, \bm{I})$ distribution.  Figure~\ref{fig:corr_vs_par_corr} illustrates the pairwise dependence structures among the correlations $r_{ij}$ and the partial correlations $\rho_{kl}$ when the pairs share (top panels) or do not share (bottom panels) common indices. When there is a shared index, the density on $(r_{12}, r_{13})$ tends to support similar values in absolute terms (the visible cross pattern), which is less apparent when there is no common index in $(r_{12}, r_{34})$. However, both distributions have most of their density on the vertices corresponding to $|r_{ij}|\approx1$.  This means that inference for all pairs of $r_{ij}$ is skewed towards jointly extreme values a priori (the univariate margin for each $r_{ij}$ is still uniform on $(-1,1)$), although this effect diminishes with an increase in the number of observations.  In contrast, pairs of partial correlations $\rho_{kl}$ exhibit no dependence structure regardless of whether or not there is a common index.  Independence is also observed between $r_{ij}$ and $\rho_{kl}$, except when both parameters have the same indices $(r_{12},\rho_{12})$ in which case they are strongly positively correlated.
\begin{figure}[t!]
    \centering
    \includegraphics[trim= 1cm 3.5cm 1.25cm 5cm,clip,width=\textwidth]{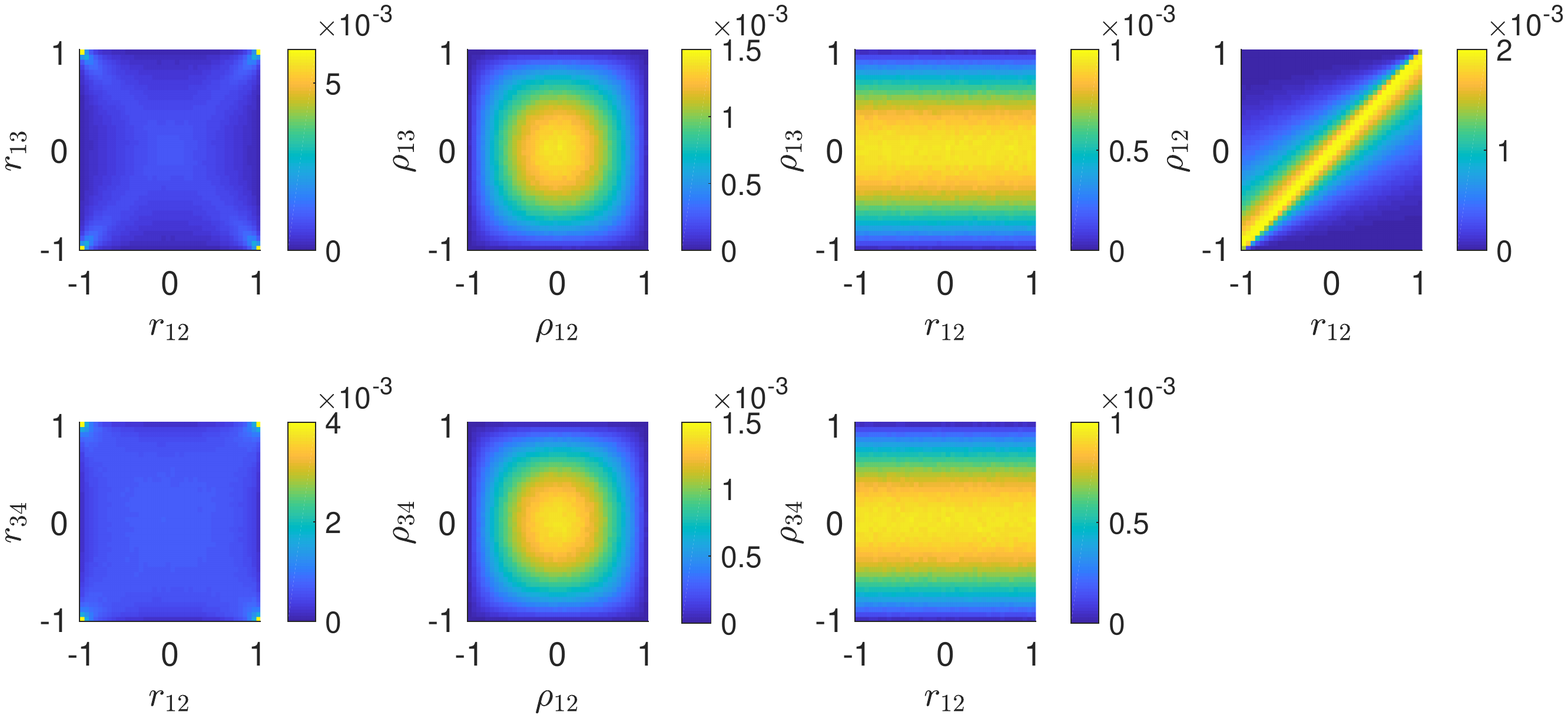}
    \caption{Bivariate density plots showing the dependence structures associated with the marginally uniform prior (\ref{eqn:marginally uniform prior}) on $\bR$ with $\nu=D+1$, for pairs of parameters sharing common indices (top panels) and without a common index (bottom panels).}
    \label{fig:corr_vs_par_corr}
\end{figure}

We now discuss related work on priors for $\bR$.  Let $\mathcal{R}^D$ be the space of all valid correlation matrices.  \cite{barnard2000modeling} also suggest a uniform prior over all correlation matrices in $\mathcal{R}^D$, which is equivalent to the LKJ prior with unit shape, as suggested by the \cite{stan2017stan}.  Note that the induced prior on the partial correlation matrix is the marginally uniform prior in (\ref{eqn:marginally uniform prior}) with $\nu=D+1$ (Figure \ref{fig:corr_vs_par_corr}).  This might not be a suitable prior for $\rho_{kl}$ since, as discussed above, this joint distribution for $\rho_{kl}$ exhibits dependence and has large mass on extreme values.  \cite{chib1998analysis} propose using a multivariate normal prior on the $r_{ij}$, with the support of the prior restricted to values of $r_{ij}$ which give a correlation matrix in $\mathcal{R}^D$, while \cite{liechty2004bayesian} introduce a mixture of normal distributions prior on $r_{ij}$ to express a priori knowledge of blocked structure in $\bR$.  However, these choices of normal priors do not imply that all marginal densities of the $r_{ij}$ are the same due to the constraints imposed on the $r_{ij}$ for the resulting $\bR$ to be in $\mathcal{R}^D$.

\subsection{Posterior sampling of \texorpdfstring{$\bR$}{bReps}} \label{subsec:posterior}

Posterior simulation for $\bR$ is challenging for two reasons: (i) the diagonal elements of $\bR$ must be 1 and, (ii) $\bR$ must be positive definite.  \cite{chib1998analysis} suggest sampling the $r_{ij}$ elements of $\bR$ in blocks using a random walk Metropolis-Hastings (RWMH) algorithm with a multivariate $t$ proposal density.  However, the resulting matrix obtained after each proposal is not guaranteed to be a valid correlation matrix in addition to the RWMH algorithm being notorious for its slow exploration of the parameter space.  Tuning the parameters of this proposal distribution also requires finding an approximate mode of the log posterior distribution and the observed Fisher information for every iteration, resulting in high computational overheads.  In the setting of hierarchical regression models, \cite{barnard2000modeling} adopt the Griddy-Gibbs sampler of \cite{ritter1992facilitating} to sample $\bR$.  Here, prior to the Gibbs step, one needs to solve a quadratic equation to determine the support for a single $r_{ij}$ (while keeping the rest fixed) which results in a valid correlation matrix.  The authors document the clear inefficiency in this sampling scheme when the prior in (\ref{eqn:marginally uniform prior}) is used due to its tendency to place more weight on the edges of $\mathcal{R}^D$ space.  Moreover, the design of drawing one $r_{ij}$ at a time becomes computationally prohibitive when $D$ is large.

\section{Efficient sampling for \texorpdfstring{$\bR$}{bReps} when using a marginally uniform prior} \label{sec:efficient}

This section describes an efficient way of sampling $\bR$ by utilising Hamiltonian dynamics \citep{duane1987hybrid}.  This involves reparameterising $\bR$ to enable sampling of parameters in an unconstrained space.  Due to the attractive properties of the marginally uniform prior in (\ref{eqn:marginally uniform prior}) with $\nu=D+1$ discussed in Section~\ref{subsec:prior}, we will use this prior hereafter.  Inference for the posterior distribution in (\ref{eqn:posterior}) can be performed using a Gibbs sampler (see Chapter 10 of \cite{greenberg2012introduction} for details). Our focus here is on the following non-standard conditional posterior distribution
\begin{equation}
\pi(\bR|\by, \by^\ast, \ba_{1:P}, \bt_{-\bR}) \propto \prod_{i=1}^P \prod_{t=1}^T \phi(\by^\ast_{it}; \bmu_{it}, \bR) \cdot p(\bR),
\label{eqn:conditional posterior}
\end{equation}
where $\bt_{-\bm{\mathcal{S}}}$ is defined as $\bt$, but excluding the parameters $\bm{\mathcal{S}}$.

\subsection{An unconstrained parameterisation} \label{subsec:unconstrained}

Because of the restrictions on sampling correlation coefficients on a confined space, we adopt the reparameterisation strategy in \cite{smith2013bayesian} which re-expresses $\bR$ via a positive definite matrix $\bSigeps$ as
\begin{equation}
\bR = \bL_{\beps}^{-1/2} \bSigeps \bL_{\beps}^{-1/2},
\label{eqn:reparameterisation}
\end{equation}
where $\bL_{\beps}=\diag(\bSigeps)$.  The covariance matrix $\bSigeps$ can then be written in terms of its Cholesky factorisation $\bSigeps=\Leps\Leps^\top,$ where $\Leps$ is a lower triangular matrix.  The diagonal elements of $\Leps$ are set to 1 so that the transformation of $\bR$ to $\Leps$ is one-to-one.  We define an operator $\vechL$ which vectorises the strict lower triangle of a matrix by row.  The unknown parameter $\vechL(\Leps)=\{L_{ij};i=2, \dotsc, D, j<i\}$ lies in $\mathbb{R}^{D(D-1)/2}$ and is therefore unconstrained.  \cite{lindstrom1988newton} also implement the Cholesky factorisation on a covariance matrix to optimise the log-likelihood function of a linear mixed effects model.  Other possible reparameterisation methods for $\bR$ include using polar coordinates \citep{rapisarda2007parameterizing} and partial autocorrelations \citep{daniels2009modeling}, but we adopt the representation in (\ref{eqn:reparameterisation}) due to its computational tractability.

By using a change of variables, we can rewrite the density function in (\ref{eqn:conditional posterior}) in terms of $\Leps$ as 
\begin{equation}
\pi(\Leps|\by, \by^\ast, \ba_{1:P}, \bt_{-\Leps}) \propto \pi(\bR|\by, \by^\ast, \ba_{1:P}, \bt_{-\bR}) \cdot |{\bf J}|,
\label{eqn:unconstrained conditional posterior}
\end{equation}
where $|{\bf J}|=|\partial \vechL(\bR) / \partial \vechL(\Leps)^\top|$ is the determinant of the Jacobian for the transformation.  We now note that for the transformation from $\bR$ to $\Leps$, the prior on lower triangular Cholesky factor $\Leps$ whose diagonal entries are all fixed as ones, given by
\begin{equation}
p(\Leps) \propto p(\bR) \cdot |{\bf J}|,
\label{eqn:prior cholesky}
\end{equation}
induces a marginally uniform prior on all $r_{ij}$ for $\nu=D+1$.

\subsection{Sampling the Cholesky factor using HMC} \label{subsec:sampling}

HMC, popularised by \cite{neal2011mcmc}, has enjoyed considerable recent interest within the statistical literature due to its ability to generate credible but distant candidate parameters for the MH algorithm, thereby reducing autocorrelation in the posterior samples.  It does so by exploiting gradient information of the log posterior density to simulate a trajectory according to physical dynamics.

Given a target distribution of interest $\pi(\vt)$, which in our case is the density in (\ref{eqn:unconstrained conditional posterior}), HMC introduces a fictitious momentum variable $\bu$ into the physical system, which is assumed to follow a $\N({\bf 0}, \bm{M})$ pseudo-prior and targets the augmented distribution
\begin{equation}
\pi(\vt, \bu) \propto \exp(-\mathcal{H}(\vt, \bu)),
\label{eqn:HMC}
\end{equation}
where $\mathcal{H}(\vt, \bu) = -\log \pi(\vt) + \frac{1}{2} \bu^\top \bm{M} ^{-1} \bu$ is termed the Hamiltonian which is made up of potential energy and kinetic energy components.  The potential energy is derived from minus the log density of $\vt$ under the target distribution while the kinetic energy is due to the movement of the momentum variable $\bu$.  The Hamiltonian system is used to describe the evolution of $\vt$ and $\bu$ over time $t$ via the differential equations
\begin{equation}
\frac{d \vt}{dt} = \frac{\partial \mathcal{H}}{\partial \bu} \quad \textrm{and} \quad \frac{d \bu}{dt} = -\frac{\partial \mathcal{H}}{\partial \vt}.
\label{eqn:hamiltonian dynamics}
\end{equation}
The dynamics in (\ref{eqn:hamiltonian dynamics}) can be implemented in practice using the leapfrog method \citep{neal2011mcmc} and discretising continuous time by a stepsize $\vareps$ so that
\begin{equation}
\begin{aligned}
\bu(t + \vareps/2) &= \bu(t) - (\vareps/2) \frac{\partial \mathcal{H}}{\partial \vt}(\vt(t)) \\
\vt(t+\vareps) &= \vt(t) + \vareps \frac{\partial \mathcal{H}}{\partial \bu}(\bu(t+\vareps/2)) \\
\bu(t+\vareps) &= \bu(t+\vareps/2) - (\vareps/2) \frac{\partial \mathcal{H}}{\partial \vt}(\vt(t+\vareps)).
\end{aligned}
\label{eqn:leapfrog}
\end{equation}
\cite{neal2011mcmc} shows that properties of the Hamiltonian such as reversibility and volume preservation are maintained under the symplectic integrator in (\ref{eqn:leapfrog}).  Proposed values $\vt'$ and $\bu'$ obtained after a trajectory length of $\mathcal{T}=n\vareps$ by iterating procedures in (\ref{eqn:leapfrog}) $n$ times are then accepted with probability $\min \{1, \exp(\mathcal{H}(\vt, \bu)-\mathcal{H}(\vt', \bu')) \}$. The invariant distribution of the Markov chain generated from the HMC algorithm is $\pi(\vt, \bu)$ and samples from $\pi(\vt)$ can be obtained by marginalising out the momentum $\bu$.

In order to implement the HMC algorithm as described above, computation of the derivatives of (\ref{eqn:unconstrained conditional posterior}) with respect to the $L_{ij}$ is required for the leapfrog update.  Lemma~\ref{lemma:derivatives for gradient} derives the expressions for these gradients.
\begin{lemma}
Let $\bm{E}_k$ denote the matrix obtained by removing column $k$ from an identity matrix $\bm{I}$.  For the parameterisation of $\bR$ in (\ref{eqn:reparameterisation}),
\begin{enumerate}[font=\normalfont, label=(\roman*)]
\item $\displaystyle \frac{\partial \bR^{-1}}{\partial L_{ij}}=-\bL_{\beps}^{1/2} \bigg( \bSigeps^{-1} \frac{\partial \bSigeps}{\partial L_{ij}} \bSigeps^{-1} + \frac{\partial \bL_{\beps}^{-1/2}}{\partial L_{ij}} \bL_{\beps}^{1/2} \bSigeps^{-1} + \bSigeps^{-1} \bL_{\beps}^{1/2} \frac{\partial \bL_{\beps}^{-1/2}}{\partial L_{ij}} \bigg) \bL_{\beps}^{1/2}$.
\item $\displaystyle \frac{\partial \log|\bR(-k;-k)|}{\partial L_{ij}}=\text{tr}\bigg(\bR^{-1}(-k;-k) \bm{E}_k^\top \frac{\partial \bR}{\partial L_{ij}} \bm{E}_k\bigg)$.
\item $\displaystyle \frac{\partial \log |\bR|}{\partial L_{ij}}= - \frac{2L_{ij}}{\sum_{k=1}^i L^2_{ik}}$.
\end{enumerate}
\label{lemma:derivatives for gradient}
\end{lemma}
\begin{proof}
Lemma 1(i) and (ii) are respectively obtained using Theorems 1 and 2 in Chapter 8 of \cite{magnus1999matrix}, by expressing $\frac{\partial \bSigeps^{-1}}{\partial L_{ij}}$ in terms of $\frac{\partial \bR^{-1}}{\partial L_{ij}}$ using the chain rule, and writing $\bR(-k;-k)$ as $\bm{E}_k^\top \bR \bm{E}_k$.  Lemma 1(iii) is straightforward by noting that $|\bR|=|\bL_{\beps}|^{-1}$ since $|\Sig_\eps|=1$ from its Cholesky decomposition.
\end{proof}

\section{A deterministic proposal distribution} \label{sec:deterministic}

Various strategies have been proposed to reduce the variability in the Monte Carlo estimate of the expectation $\E[f(\vt)]$ of a scalar function $f$ of parameter $\vt$ with respect to some posterior distribution $\pi(\vt)$, including the Rao-Blackwellisation \citep{robert2004monte} and the control variates \citep{dellaportas2012control, oates2017control}.  These techniques produce an efficient estimator of $\E[f(\vt)]$ based on sampled $\vt$ generated from an MCMC sampler.

Here, we focus on a particular class of methods which integrate variance reduction techniques dynamically within an MCMC sampling algorithm.  Let $\vt=(\vartheta_1, \dotsc, \vartheta_n)^\top$ be a parameter vector with normal full conditional distributions $\vartheta_i|\vt_{-i} \sim \N(\mu_i,\s_i^2)$, where the conditional mean $\mu_i$ and the conditional variance $\s^2_i$ may depend on $\vt_{-i} = \{\vartheta_j:j=1, \dotsc,n, j \not = i\}$.  \cite{adler1981over} and \cite{barone1990improving} introduce an over-relaxation method where the update on $\vt$ is performed using Gibbs sampling, and where the new value $\vartheta'_i$ for each margin of $\vt$ is generated as 
\begin{equation}
    \vartheta'_i = (1+\k)\mu_i - \k \vartheta_i + u \s_i \sqrt{1-\k^2}, \quad i=1,\dotsc,n,
    \label{eqn:over-relaxation}
\end{equation}
with $u\sim \mathcal{N}(0,1)$ being a standard normal random variable. Equation (\ref{eqn:over-relaxation}) allows for the introduction of dependence between successive samples via the constant antithetic parameter $\k$, which is required to be in the open interval $(-1,1)$ so that the Markov chain is ergodic and produces $\pi(\vt)$ as its stationary distribution.  This scheme is exactly the conventional Gibbs sampler when $\k=0$. Variance reduction in estimating $\E[f(\vt)]$ is achieved through the antithetic variable method \citep{hammersley1956new} by setting $\k>0$ so that the estimation bias in the previous sample is corrected in the opposite direction.   The rate of convergence for the over-relaxation method in (\ref{eqn:over-relaxation}) is studied in \cite{barone1990improving}, while \cite{green1992metropolis} establish that the asymptotic variance of the estimator for $\E[f(\vt)]$ using this strategy for linear $f$ is proportional to $\frac{1-\k}{1+\k}$.

The inefficiency of an MCMC sampler in estimating $\E[f(\vt)]$ is usually measured by the integrated autocorrelation time \citep{roberts2009examples}, which is defined as
\begin{equation*}
\textrm{IACT}_f = 1 + \sum_{j=1}^\infty \rho_{j,f},
\end{equation*}
where $\rho_{j,f}$ is the lag $j$ autocorrelation function of the MCMC iterates of $f(\vt)$ after convergence.  Alternatively, one can measure the efficiency of the sampler by computing the effective sample size per MCMC iteration, which by definition is the reciprocal of the IACT.  A small value of the IACT is desirable in practice as it indicates that the Markov chain mixes well.  Motivated by the over-relaxation sampler and noting that the IACT can be less than 1 if some of the autocorrelations are negative, in which case a Monte Carlo estimator constructed is super-efficient, we introduce into the MH literature a deterministic design of the proposal distribution for $\vt$
\begin{equation}
q(\vt'|\vt) = \d_{\psi(\vt)}(\vt'),
\label{eqn:deterministic proposal}
\end{equation}
where $\psi$ is a mapping function which introduces negative correlation between samples and $\d_{\psi(\vt)}$ is the Dirac delta function at $\psi(\vt)$.  In this case, the MH acceptance probability involves the ratio of $\pi(\vt)$ evaluated at $\vt'$ and $\vt$.

When $\pi(\vt)$ is a normal distribution, we propose setting
\begin{equation}
\psi(\vt) = 2 \bmu_{\vt} - \vt,
\label{eqn:antithetic}
\end{equation}
where $\bmu_{\vt}$ is the mean of $\pi(\vt)$.  It is clear that (\ref{eqn:antithetic}) represents an example of the antithetic variable with perfect negative correlation, and also an instance of the over-relaxation method in (\ref{eqn:over-relaxation}) with $\k=1$, which is outside the range of values for which the Markov chain is ergodic.  Symmetry of the normal density gives $\pi(\vt')=\pi(\vt)$, which in turn translates to an acceptance probability of one.  Clearly, our proposed antithetic sampling will only yield an ergodic Markov chain when it is coupled with stochastic simulation of additional parameters that affect the value of the deterministic proposal $\psi(\vt)$, in particular $\bmu_{\vt}$.  Under this condition, the value of $\bmu_{\vt}$ changes in every iteration of the update and this drives the exploration of $\vt$ in the parameter space.  Furthermore, the dependence between $\vt$ and other model parameters prevents exact periodicity from occurring, and thus the Markov chain is aperiodic.

The conditional posterior distribution of the random effects $\ba_{1:P}$ in our MVP model is normal and likewise for the regression parameters $\bb$ when using a conjugate prior.  Therefore, we can employ the antithetic sampling method in (\ref{eqn:antithetic}) to improve the IACTs of $\ba_{1:P}$ and $\bb$.  In fact, antithetic sampling of normal random variables can also be understood in terms of a HMC update.  Suppose that $\vt \sim \N(\bmu_{\vt}, \Sig_{\vt})$, and the prior on the momentum variable $\bu$ is chosen as $\N({\bf 0}, \Sig^{-1}_{\vt})$.  \cite{pakman2014exact} show that the resulting Hamiltonian system can be solved analytically, with solution given by
\begin{equation}
\vt(t) = \bmu_{\vt} + \Sig_{\vt} \bu(0) \sin(t) + (\vt(0)-\bmu_{\vt}) \cos(t),
\label{eqn:solution for Hamilton}
\end{equation}
which is a linear combination of $\bmu_{\vt}$, the initial value $\vt(0)$ of $\vt$ and the initial momentum $\bu(0)$.  Note that (\ref{eqn:solution for Hamilton}) is a multivariate generalisation of (\ref{eqn:over-relaxation}) with $t=\textrm{cos}^{-1}(-\k)$.  Equation~(\ref{eqn:solution for Hamilton}) is thus equivalent to the antithetic sampler in (\ref{eqn:antithetic}) when setting $t=\pi$ radians.  Since there is no approximation error in the Hamiltonian dynamics for a normal distribution, an MH accept-reject step is not required in the HMC sampler, and the proposed value of $\vt$ will always be accepted.  This equivalence relation was first observed by \cite{pakman2014exact}, but was not particularly useful in their framework of sampling from a truncated multivariate normal distribution.  Our proposal for antithetic sampling is different from theirs in the sense that it is entirely deterministic, and we choose $t=\pi$ radians to induce a perfect negative proposal correlation.  \cite{pakman2014exact}, on the other hand,  suggest setting $t=\frac{\pi}{2}$ radians, which is equivalent to drawing a fresh sample from a random number generator when it is applied to the setting of a normal distribution.  We refer to this approach as the independent sampler hereafter.

So far, our discussion has mainly focused on normal $\pi(\vt)$.  This is because an analytic solution to the Hamiltonian system is only available for a normal distribution.  
It is possible to extend the proposed antithetic sampler to more general distributions by obtaining an approximation of $\bmu_{\vt}$ in order to propose a new value of $\vt$, and then accept or reject the proposal in an MH algorithm to target the true $\pi(\vt)$, as suggested in \cite{green1992metropolis}.  However, the application of this generalisation and its variants (e.g.~\cite{creutz1987overrelaxation}) is somewhat limited due to high rejection rates in the accept-reject step \citep{neal1998suppressing}. In this case, the HMC algorithm provides a way to overcome this shortcoming.

\section{Simulation studies} \label{sec:simulation}

We now study the efficiency of the antithetic variable technique described in Section~\ref{sec:deterministic}.  Two examples are presented.  The first examines the antithetic sampler in a more general setting, while the second is specific to the application in Section~\ref{sec:application}.  Reported IACT values of the parameters are computed using the \verb|coda| package \citep{plummer2006coda} in \verb|R|.
\vspace{\baselineskip}

\noindent
{\bf Example 1. }The stationary distribution $\pi(\bt)$ is specified as a bivariate normal distribution with high correlation (0.99) between the variables.  We investigate the performance of three sampling schemes - the independent sampler, the over-relaxation algorithm with $\k=0.9$, and a coupling of the over-relaxation algorithm (on the first margin) with the antithetic sampler (on the second margin).  Note that this coupling strategy introduces stochasticity into the antithetic sampler, which is essential to produce an ergodic Markov chain.  The samplers are each run for 10\,000 iterations from the same initialised value $(2,2)$, and the update on each margin is performed conditional on the other.  Figure~\ref{fig:correlated normal} illustrates the trajectories of the first 50 samples generated.  Exploration of the target space is reduced to a random walk under the independent sampler.  In contrast, the other two samplers move between different contours of the density and explore the full support of the distribution in an elliptical manner, thereby reducing the IACT significantly.  The IACT decreases further when the over-relaxation algorithm on the second margin is replaced by antithetic sampling.  In this analysis, the mixing of both margins is improved by a factor of 1.75.
\begin{figure}[h!]
    \centering
    \includegraphics[trim= 0.5cm 1cm 6.5cm 9.5cm,clip,width=0.6\textwidth]{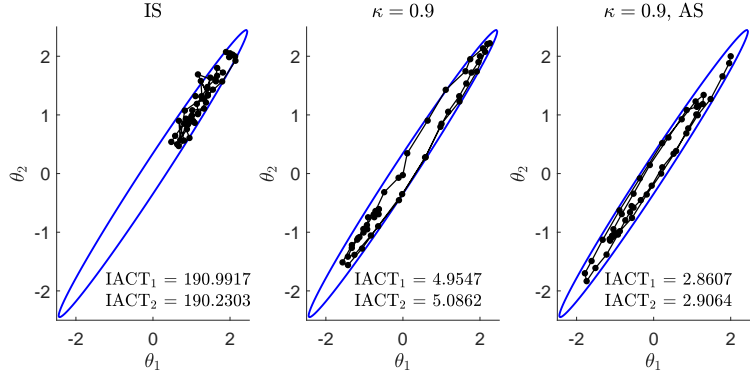}
    \caption{Trajectories of the first 50 samples generated from the independent sampler (left), the over-relaxation algorithm with $\k=0.9$ (middle), and the over-relaxation algorithm coupled with the antithetic sampler (right).  The blue solid lines represent the 95\% confidence region of the bivariate normal distribution.}
    \label{fig:correlated normal}
\end{figure}
\vspace{\baselineskip}

\noindent
{\bf Example 2. }A simulated dataset is generated following the MVP model given in (\ref{eqn:MVP})--(\ref{eqn:latent}), with $D=8, P=162, T=16$ and values of the parameters $\bt=(\bb, \bR, \Sig_{\ba})$ set to be the posterior mean estimates of the parameters in Model 1 of the female contraceptive product analysis of Section~\ref{sec:application}.  To avoid hand-tuning the stepsize $\vareps$ and the trajectory length $\mathcal{T}$ for the HMC update of $\Leps$, we utilise the No-U-Turn Sampler (NUTS) with the dual averaging scheme of \cite{hoffman2014no}.  We use the following non-informative prior distributions: $\bb \sim \mathcal{N}({\bf 0}, 100\bm{I})$, $\Sig_{\ba} \sim \mathcal{IW}(9, \bm{I})$ and the prior distribution on the lower triangular Cholesky factor $\Leps$ given in (\ref{eqn:prior cholesky}).  The sampling scheme is run for 30\,000 iterations, with the first 5\,000 samples discarded as burn-in.  Appendix~\ref{app:sampling scheme} details the Gibbs sampling scheme.

Figure~\ref{fig:autocorrelation} compares graphically the marginal posterior densities and sample autocorrelations of randomly sampled random effects $\ba_{1:P}$ and the regression parameter $\bb$ between independent and antithetic sampling.  Despite the absence of a stochastic component in the updates of $\ba_{1:P}$ and $\bb$, the kernel density plots of these parameters indicate that the coupling of a stochastic MCMC scheme for the remaining parameters with the antithetic variable technique gives the same posterior distributions as those under independent sampling.  The autocorrelation plots show that the samples generated from antithetic sampling have positive dependence with a higher rate of decay over the number of lags, thereby demonstrating the superior mixing of the Markov chain.  The IACT values of the randomly sampled parameters are significantly lower, with improvement factors of 3.72 and 2.10 observed for $\a_{3,80}$ and $\b_{182}$ respectively.  The box plot showing the distribution of the IACT values of $\bm{\a}_{1:P}$ also indicates that some of these parameters are super-efficient.  Furthermore, the log IACT ratios of the independent sampler compared to the antithetic sampler are well above 0, suggesting that all $\ba_{1:P}$ and $\bb$ parameters experience efficiency gains.  Although perfect negative correlation is induced between successive samples by the deterministic proposal, this does not necessarily translate to an equivalent autocorrelation in the posterior samples.  Rather, the negative relationship is used to reduce the magnitude of positive autocorrelation present in the MCMC samples.  Note that convergence to the posterior distribution might be slow for poorly initialised values under antithetic sampling so we suggest using independent sampling during the burn-in period and later switching to the deterministic proposal.
\begin{figure}[t!]
\onehalfspacing
\centering
\includegraphics[trim= 0.2cm 0.75cm 1.25cm 0.75cm,clip,width=0.7\textwidth]{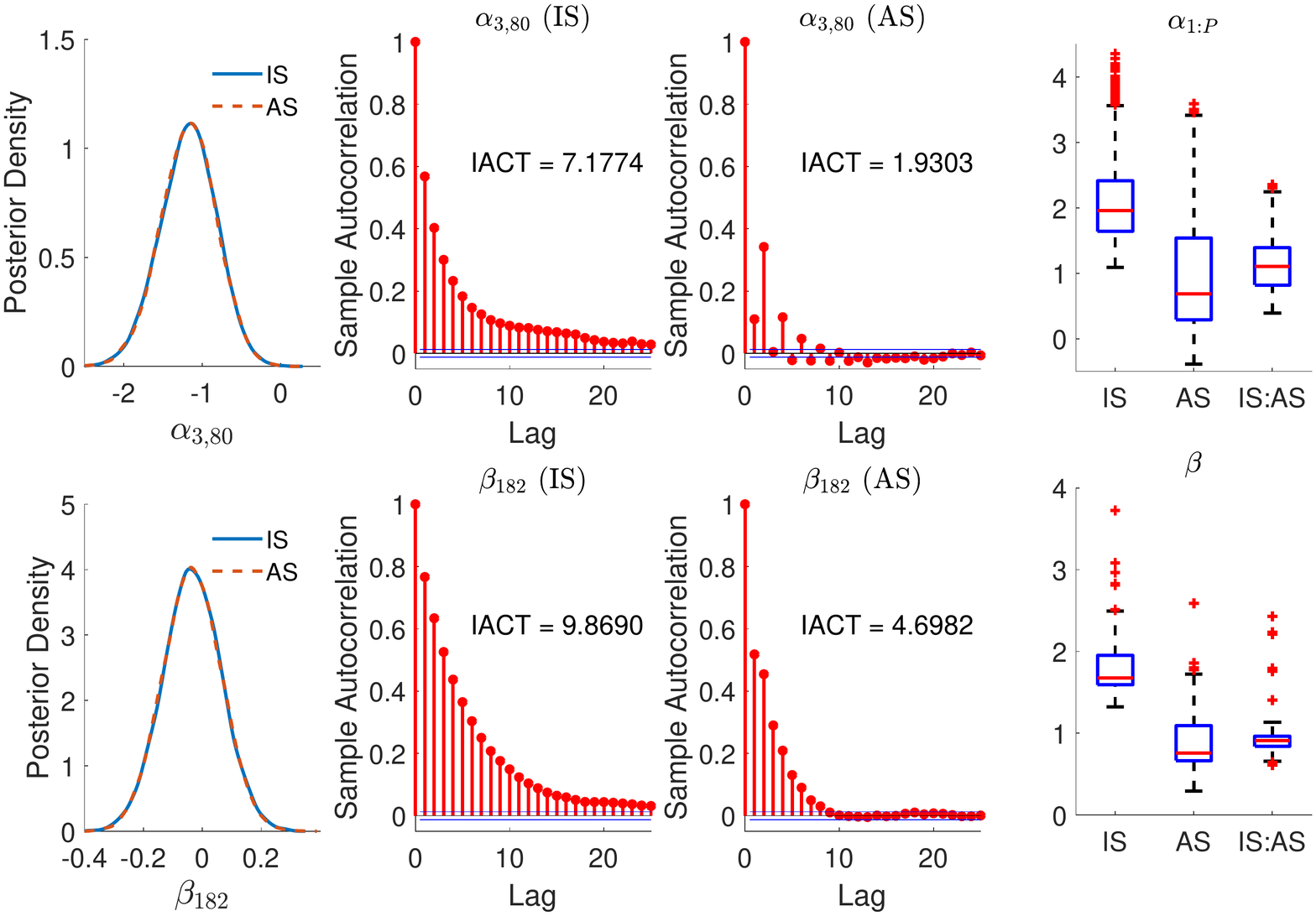}
\caption{Marginal posterior densities of a randomly selected random effects term (top panel) and regression coefficient (bottom panel), and their sample autocorrelation plots under independent sampling (IS) and antithetic sampling (AS).  Rightmost column gives the distributions of the IACT values and the element-wise IACT ratios of IS to AS for $\ba_{1:P}$ and $\bb$ on the log scale.}
\label{fig:autocorrelation}
\end{figure}

The remaining simulation experiments investigate the performance of the MVP model in the context of recovering the true parameters of the data generating process under different specifications of prior distribution on $\bt$.  We use the posterior root-mean-square error (RMSE) defined by
\begin{equation}
\textrm{RMSE}(\bt) = \sqrt{\frac{1}{N}\sum_{j=1}^N (\bt^{[j]}-\bt_{\textrm{true}})^2},
\label{eqn:RMSE}
\end{equation}
as the performance measure, where $\bt^{[j]}$ is the $j$-th iterate from the $N$ posterior samples and $\bt_{\textrm{true}}$ is the true value of $\bt$.  
The measure in (\ref{eqn:RMSE}) is defined for univariate $\bt$.  For a multivariate $\bt$, the posterior RMSE is calculated for each margin of $\bt$.  All the results shown are based on 1\,000 different replicate sets of simulated data with the same true parameter values.

We first consider the conditionally conjugate hierarchical inverse-Wishart $\mathcal{HIW}(\l, \bm{A})$ prior of \cite{huang2013simple} with degrees of freedom $\l$ and positive scale parameter $\bm{A}=(A_1, \dotsc, A_D)^\top$ as an alternative to the inverse-Wishart prior on the $D \times D$ covariance matrix $\Sig_{\ba}$,
\begin{equation*}
\begin{gathered}
\Sig_{\ba}|a_1,\dotsc,a_D \sim \mathcal{IW}\Bigg(\l + D - 1, 2 \l\diag\bigg(\frac{1}{a_1}, \dotsc, \frac{1}{a_D}\bigg)\Bigg), \\
a_i \stackrel{iid}{\sim} \mathcal{IG}(0.5, A^{-2}_i), \quad i=1,\dotsc, D,
\end{gathered}
\end{equation*}
where $\mathcal{IG}(a, b)$ is an inverse-Gamma distribution with shape $a$ and scale $b$.  The marginal prior of the standard deviation in $\Sig_{\ba}$ is a half-$t(\l, A_i)$ distribution, as suggested in \cite{gelman2006prior}.  In the simulation, we select $\l=2$ and choose a weakly informative scale parameter whereby $A_1 = A_2 = 0.23$ and $A_3 = \cdots = A_8 = 0.46$ so that approximately 95\% of the half-$t$ density is below 1 and 2 respectively.  This specification is relevant to the real data application in Section~\ref{sec:application}, where our prior belief is that the variability in the tendency of GPs to discuss pill contraceptives is lower compared to non-pill alternatives.  In contrast, the inverse-Wishart prior assumes the same variability for all variance parameters $\s^2_{\a_i}$ in $\Sig_{\ba}$.  Figure~\ref{fig:IW vs HIW} shows the distribution of the average RMSE ratio of each type of parameter in $\Sig_{\ba}$, based on 1\,000 replicate simulations, for the hierarchical inverse-Wishart prior versus the inverse-Wishart prior. Although the hierarchical inverse-Wishart prior is flexible enough to specify different strengths of prior on each $\s^2_{\a_i}$, Figure~\ref{fig:IW vs HIW} shows that in this case its performance is similar to the more restrictive inverse-Wishart prior.  This result is somewhat unsurprising considering that the estimated $\s^2_{\a_i}$ in the application example are more or less similar across the different contraceptive products (see Appendix~\ref{app:covariance}).  The distributions for the posterior RMSE ratio of the correlation coefficients and the partial correlations are concentrated around 1 since both the hierarchical inverse-Wishart prior with $\l=2$ and the inverse-Wishart prior with $D+1$ degrees of freedom and scale matrix $\bm{I}$ induce the same marginally uniform prior, i.e.~(\ref{eqn:marginally uniform prior}) with $\nu = D+1$, on the resulting correlation matrix $\bm{R}_{\ba}$, which in turn gives the same implied LKJ distribution on the partial correlations.
\begin{figure}[t!]
\centering
\begin{subfigure}{.475\textwidth}
  \centering
  \includegraphics[trim= 2.25cm 1.75cm 12.25cm 11.5cm,clip,width=0.8\textwidth]{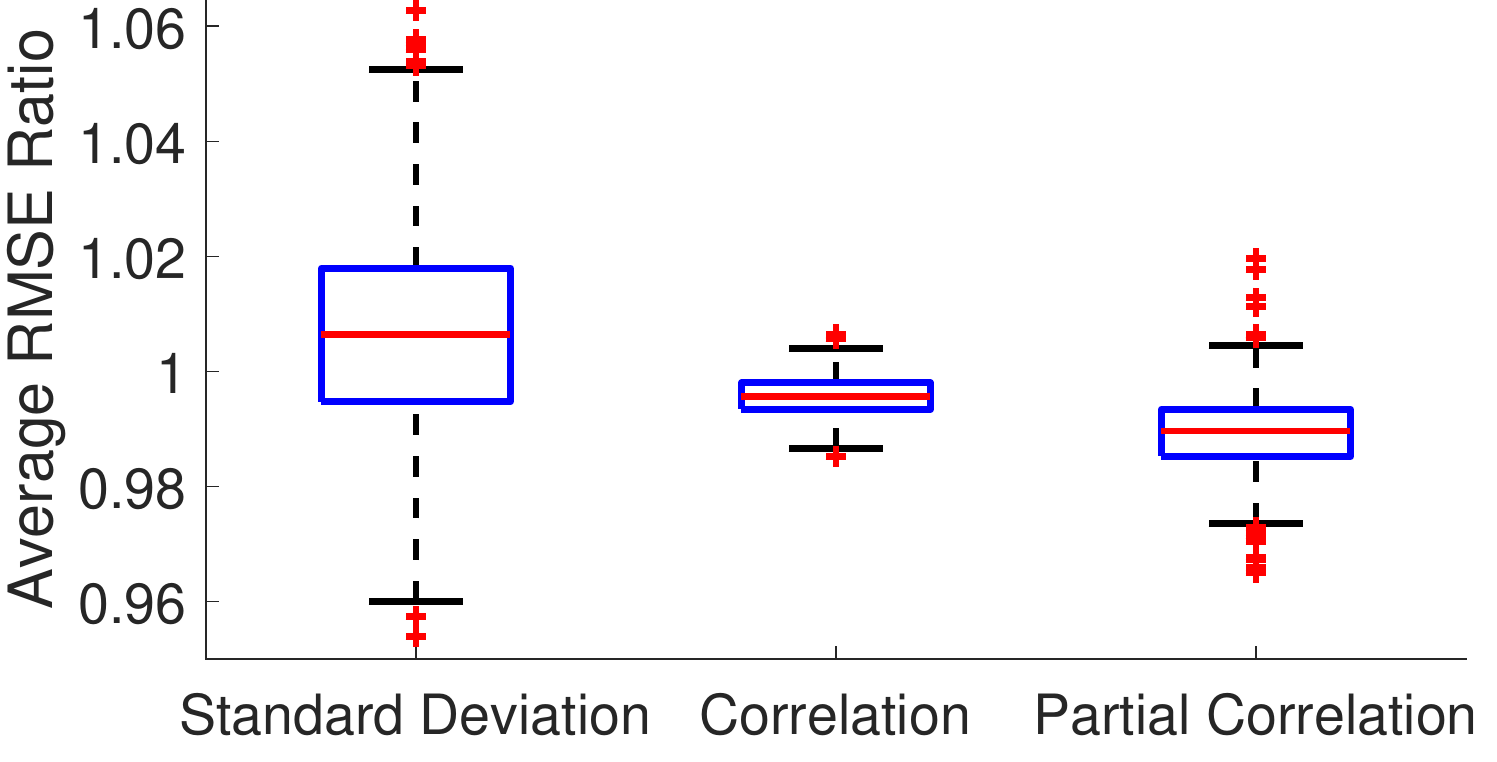}
  \caption{Hierarchical inverse-Wishart prior versus inverse-Wishart prior on $\Sig_{\ba}$.}
\label{fig:IW vs HIW}
\end{subfigure}
\hfill
\begin{subfigure}{.475\textwidth}
  \centering
  \includegraphics[trim= 2.25cm 1.75cm 12.25cm 11.5cm,clip,width=0.8\textwidth]{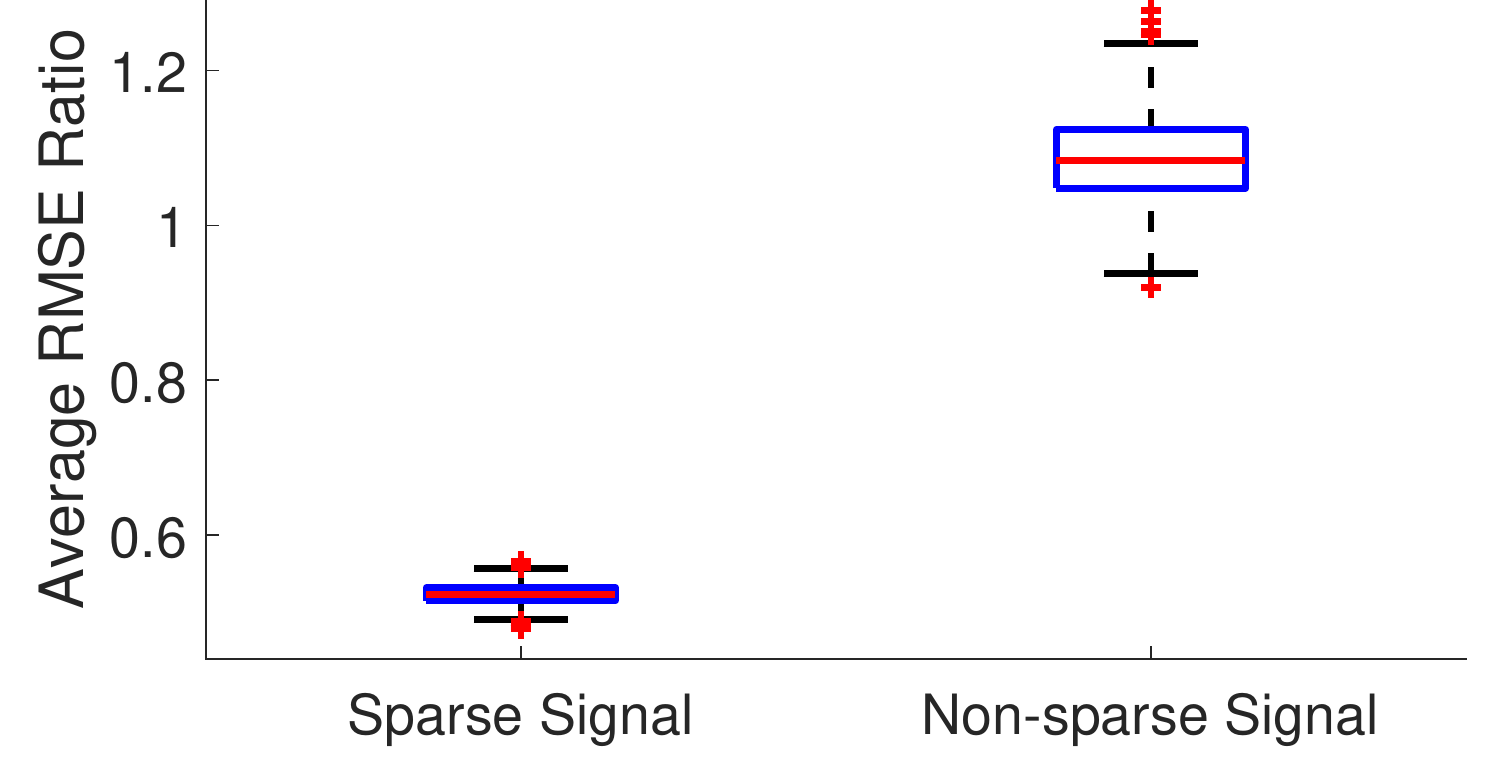}
  \caption{Horseshoe shrinkage prior versus normal prior on $\bb$.}
\label{fig:horseshoe vs normal}
\end{subfigure}
\caption{Distributions of the average posterior RMSE ratio of all parameters in (a) $\Sig_{\ba}$ or (b) $\bb$, based on 1\,000 replicate analyses, under different prior choices. (a) Standard deviations, correlations and partial correlations for parameters in $\Sig_{\ba}$ for the hierarchical inverse-Wishart prior versus the inverse-Wishart prior on $\Sig_{\ba}$. (b) Sparse regression coefficients $\beta_i=0$ and non-sparse coefficients $\beta_i\neq 0$ for the horseshoe prior versus the $\mathcal{N}(\boldsymbol{0},100\boldsymbol{I})$ prior on $\bb$.
}
\label{fig:prior beta and sigma}
\end{figure}

To identify sparse signals (coefficients which are significant) in the regression parameter $\bb$, we employ the horseshoe shrinkage prior \citep{carvalho2010horseshoe} given by
\begin{equation*}
\b_i|\l_i, \tau \sim \mathcal{N}(0, \tau^2 \l_i^2), \quad \l_i \sim \mathcal{C}^+(0,1), \quad \tau \sim \mathcal{C}^+(0,1),
\end{equation*}
where $\mathcal{C}^+(0,1)$ is a half-Cauchy distribution with location 0 and scale 1 restricted to positive support.  The simulation is carried out by setting 75\% of the smallest non-intercept regression coefficients (in absolute value) in $\bb$ to 0, from which we generate the simulated datasets.  We model the prior on each intercept separately by a flat $\mathcal{N}(0, 100)$ distribution to avoid heavily penalising these parameters.  Gibbs sampling from the posterior distribution of $\bb$ is implemented by adopting the latent variable formulation in \cite{makalic2016simple}.  Figure~\ref{fig:horseshoe vs normal} displays the results of comparing this prior specification for $\bb$ to a $\mathcal{N}({\bf 0}, 100\bm{I})$ prior, again in terms of the average RMSE ratio over all regression parameters.  The horseshoe prior performs as well as the $\mathcal{N}({\bf 0}, 100\bm{I})$ prior on non-zero entries of $\bb$, although the variability in the RMSE ratio is large.  On the other hand, the horseshoe prior outperforms the normal prior for those parameters whose true values are zero, reducing the RMSE by half.  This occurs as the horseshoe prior places a greater density around zero, which results in a more concentrated posterior distribution for parameters which are truly zero.  Therefore, it is an attractive default option when we expect sparsity in the regression parameters, as is the case for our analysis of the characteristics affecting the decision-making behaviour of GPs in the next section.

\section{Discussion of female contraceptive products by Australian GPs} \label{sec:application}
\subsection{Background and aims of study} \label{subsec:background}

In order to study the decision-making behaviour of Australian GPs, we obtain data from \cite{fiebig2017consideration} who design a stated preference experiment in which GPs are asked to select the contraceptive products that they would consider discussing with hypothetical female patients.  The GPs evaluate a sequence of vignettes where patients are defined in terms of socio-economic and clinical characteristics that are varied as part of the experimental design.  Table~\ref{table:attributes} in Appendix~\ref{app:attributes} contains the attributes of the patients with a description for each level of the categorical variables.  The GPs choose from a set of 9 products that they would discuss with the patient before deciding upon their most preferred product to be subsequently prescribed to the patient.  A sample of 162 GPs participated in the experiment where each subject makes choices for 16 different patients, resulting in 2\,592 observations.  The following covariate information is collected on the GPs themselves: age, gender, whether they are registered as a Fellow of the Royal Australian College of GPs, whether they have a certificate in family planning, whether they are an Australian medical graduate, whether their location of practice is in an urban area and whether they bulk-bill patients.  Analysis of this panel data is based on the set of binary outcomes as to whether or not to discuss each of the contraceptive products.  Due to low occurrences for the prescription of the hormonal patch which was yet to be released in the Australian market, we removed this product from the dataset leaving observations on the 8 remaining products.

The experiment is designed to mimic the choice problem faced by GPs in a consultation where they need to match a product with a particular patient.  In characterising such a decision problem, \cite{frank2007custom} distinguish between ``custom-made'' and ``ready-to-wear'' (or norm-based) choices.  A custom-made choice involves the GP undertaking a careful evaluation of the patient and then matching her to an appropriate product.  However, as new products are introduced, GPs face considerable costs in the process of gaining the knowledge and expertise required to discuss and prescribe these products.  This is particularly the case when more familiar products are available even though they may be somewhat inferior to the new products; an especially salient situation in the market for contraceptive products.  In such cases, some GPs will tend to adopt norms (here particular products) that work well for a broad class of patients and to place less weight on certain patient attributes that would indicate a different product that is potentially a better match.

Particular interest is in the dependence between the products.  That is, which products tend to be discussed together and which tend to form distinct clusters.  If GPs pursue custom-made strategies, then a considerable portion of the dependence between products will be explained by the attributes of the patient.  Conditional on the observable features of the patient and characteristics of the GPs, remaining dependencies will reflect the relationship between unobservables related to evaluations of the suitability of certain products for a particular patient, and how individual GP's product effects are correlated across products.  The proposed model is designed to capture these forms of heterogeneity and will permit a detailed analysis of the choices.

The prevalence of ready-to-wear choices is one possible explanation for the relatively low uptake of long acting reversible contraceptive (LARC) methods in Australia \citep{black2013australian}.  LARC methods are contraceptives that are administered less frequently than monthly and include hormonal implants, intrauterine contraception (IUC), both hormonal and copper-bearing, and contraceptive injections.  There is increasing support for the greater use of these more effective methods to reduce unintended pregnancies and abortion rates.  In our analysis below, we will use the model to explore a case where there is no clinical reason why at least one of these LARC methods should not be considered for discussion by GPs.  For ease of presentation, we will use the subscripts in Table~\ref{table:products} to denote the products.
\begin{table}[ht!]
\centering
\onehalfspacing
\begin{tabular}{|c|l|}
\hline
Subscript & \multicolumn{1}{c|}{Product} \\ \hline
1 & Combined pill \\
2 & Mini-pill \\
\cellcolor{gray!25}3 & \cellcolor{gray!25}Hormonal injection \\
\cellcolor{gray!25}4 & \cellcolor{gray!25}Hormonal implant \\
\cellcolor{gray!25}5 & \cellcolor{gray!25}Hormonal IUD \\
6 & Vaginal ring \\
\cellcolor{gray!25}7 & \cellcolor{gray!25}Copper IUD \\
8 & Condom \\ \hline
\end{tabular}
\caption{Correspondence of parameter subscripts to each female contraceptive product.  Long acting reversible contraceptive methods are shown in grey.}
\label{table:products}
\end{table}

\subsection{Analysis and results} \label{subsec:discussion}

We consider two different models for the data:
\begin{align}
\text{Model 1: }& \by^\ast_{it} = \ba_i + \bB \bx_{it} + \beps_{it}, \label{eqn:model 1}\\
\text{Model 2: }& \by^\ast_{it} = \ba_i + \bB \bx_{it} + \bm{C}  \bz_i + \beps_{it}, \label{eqn:model 2}
\end{align}
for $i=1,\dotsc,P=162$ GPs and $t=1,\dotsc,T=16$ patients.  Here $\ba_i$ and $\bm{C}  \bz_i$ respectively represent GP-specific random and fixed effects with $\bz_i$ being a vector of GP characteristics, and $\bB \bx_{it}$ represents fixed effects of the patient.  We select a horseshoe prior on $\bb=\textrm{vec}(\bB)$ and model the covariance matrix $\Sig_{\ba}$ of the random effects by the $\mathcal{HIW}(2,\bm{A})$ prior in Section~\ref{sec:simulation} where $\bm{A}=(0.23, 0.23, 0.46, \dotsc, 0.46)^\top$.  The scale is chosen to express the prior information that the variances of the random effects are expected to be small, with those for the pill products being less variable compared to the non-pill alternatives.  The difference between these two models is the presence of the GP-specific fixed effects in Model 2, which explain some of the relationships in the random effects of Model 1.  Let $\bm{X}=(X_1, \dotsc, X_D)^\top$ be a vector of normal random variables with covariance matrix given by $\Sig_{\bm{X}}$.  Recall that $X_i$ and $X_j$ are conditionally independent given the other random variables if the $(i,j)$-th entry of the precision matrix $\Sig^{-1}_{\bm{X}}$ is zero.

Figures~\ref{fig:correlation graph} and \ref{fig:random effect graph} give graphical summaries of the posterior distribution of the dependence structures of the latent variable $\by^\ast_{it}$ conditional on $\ba_i$ and $\bx_{it}$ (as well as $\bm{z}_i$ for Model 2), and the random effects $\ba_i$ respectively.  All graphs are obtained by computing the 95\% credible interval of the posterior distribution for each entry of $\bR^{-1}$ and $\Sig^{-1}_{\ba}$, where an edge is formed between two nodes if the credible interval does not include 0. The absence of an edge between any two nodes indicates a potential conditional independence between the two variables given the rest.  The dependence structures associated with the latent variables are the same for both models.  This supports the use of the MVP model in order to capture the complex dependencies between different products that would otherwise be ignored in separate univariate analyses on each product.  
\begin{figure}[ht!]
\centering
\begin{tikzpicture}[>=stealth',shorten >=1pt,auto,thick,main node/.style={circle,draw}]
	\draw (67.5:2cm)   node[main node, scale=0.9] (1) {$y^\ast_1$};
    \draw (22.5:2cm)   node[main node, scale=0.9] (2) {$y^\ast_2$};
    \draw (337.5:2cm)  node[main node, scale=0.9] (3) {$y^\ast_3$};
    \draw (292.5:2cm)  node[main node, scale=0.9] (4) {$y^\ast_4$};
    \draw (247.5:2cm)  node[main node, scale=0.9] (5) {$y^\ast_5$};
    \draw (202.5:2cm)  node[main node, scale=0.9] (6) {$y^\ast_6$};
    \draw (157.5:2cm)  node[main node, scale=0.9] (7) {$y^\ast_7$};
    \draw (112.5:2cm)  node[main node, scale=0.9] (8) {$y^\ast_8$};
    \path[every node/.style={font=\sffamily\small}]
    (1) edge[red, line width=0.08cm]  node [left] {} (5)
        edge[blue, line width=0.16cm] node [left] {} (6)
    (2) edge[blue, line width=0.055cm] node [left] {} (3)
        edge[red, line width=0.091cm]  node [left] {} (6)
    (3) edge[blue, line width=0.055cm] node [left] {} (2)
    	edge[blue, line width=0.2cm] node [left] {} (4)
    (4) edge[blue, line width=0.2cm] node [left] {} (3)
    	edge[blue, line width=0.082cm] node [left] {} (5)
    (5) edge[red, line width=0.08cm]  node [left] {} (1)
    	edge[blue, line width=0.082cm] node [left] {} (4)
		edge[blue, line width=0.176cm] node [left] {} (7)
        edge[red, line width=0.026cm]  node [left] {} (8)
    (6) edge[blue, line width=0.16cm] node [left] {} (1)
    	edge[red, line width=0.091cm]  node [left] {} (2)
    	edge[blue, line width=0.062cm] node [left] {} (8)
    (7) edge[blue, line width=0.176cm] node [left] {} (5)
    	edge[blue, line width=0.07cm] node [left] {} (8)
    (8) edge[red, line width=0.026cm]  node [left] {} (5)
    	edge[blue, line width=0.062cm] node [left] {} (6)
        edge[blue, line width=0.07cm] node [left] {} (7);
\end{tikzpicture}
\caption{Graphical model illustrating substantial dependence structure of the latent variables $\by^\ast$ conditional on the random effects and the covariates in both Model 1 and 2.  Edges between $y^\ast_i$ and $y^\ast_j$ are included if the 95\% credible interval of the marginal posterior distribution of the $(i,j)$-th entry of $\bR^{-1}$ does not contain 0.  Blue edges represent positive dependence while red edges represent negative dependence.  The thickness of the edges is proportional to the strength of the dependence.}
\label{fig:correlation graph}
\vspace{10pt}
\begin{subfigure}[b]{0.45\textwidth}
\begin{tikzpicture}[>=stealth',shorten >=1pt,auto,thick,main node/.style={circle,draw}]
	\draw (67.5:2cm)  node[main node, scale=0.9] (1) {$\a_1$};
    \draw (22.5:2cm)  node[main node, scale=0.9] (2) {$\a_2$};
    \draw (337.5:2cm) node[main node, scale=0.9] (3) {$\a_3$};
    \draw (292.5:2cm) node[main node, scale=0.9] (4) {$\a_4$};
    \draw (247.5:2cm) node[main node, scale=0.9] (5) {$\a_5$};
    \draw (202.5:2cm) node[main node, scale=0.9] (6) {$\a_6$};
    \draw (157.5:2cm) node[main node, scale=0.9] (7) {$\a_7$};
    \draw (112.5:2cm) node[main node, scale=0.9] (8) {$\a_8$};
    \path[every node/.style={font=\sffamily\small}]
    (1) edge[blue, line width=0.17cm] node [left] {} (2)
    	edge[blue, line width=0.114cm] node [left] {} (3)
    (2) edge[blue, line width=0.17cm] node [left] {} (1)
    (3) edge[blue, line width=0.114cm] node [left] {} (1)
    	edge[blue, line width=0.074cm] node [left] {} (7)
    (4) edge[blue, line width=0.141cm] node [left] {} (5)
    	edge[red, line width=0.063cm]  node [left] {} (7)
    (5) edge[blue, line width=0.141cm] node [left] {} (4)
    (6) edge[blue, line width=0.098cm] node [left] {} (8)
    (7) edge[blue, line width=0.074cm] node [left] {} (3)
    	edge[red, line width=0.063cm]  node [left] {} (4)
    	edge[blue, line width=0.068cm] node [left] {} (8)
    (8) edge[blue, line width=0.098cm] node [left] {} (6)
    	edge[blue, line width=0.068cm] node [left] {} (7);
\end{tikzpicture}
\centering
\caption{Model 1.}
\label{fig:model 1 graph}
\end{subfigure}
\centering
\begin{subfigure}[b]{0.45\textwidth}
\begin{tikzpicture}[>=stealth',shorten >=1pt,auto,thick,main node/.style={circle,draw}]
	\draw (67.5:2cm)  node[main node, scale=0.9] (1) {$\a_1$};
    \draw (22.5:2cm)  node[main node, scale=0.9] (2) {$\a_2$};
    \draw (337.5:2cm) node[main node, scale=0.9] (3) {$\a_3$};
    \draw (292.5:2cm) node[main node, scale=0.9] (4) {$\a_4$};
    \draw (247.5:2cm) node[main node, scale=0.9] (5) {$\a_5$};
    \draw (202.5:2cm) node[main node, scale=0.9] (6) {$\a_6$};
    \draw (157.5:2cm) node[main node, scale=0.9] (7) {$\a_7$};
    \draw (112.5:2cm) node[main node, scale=0.9] (8) {$\a_8$};
    \path[every node/.style={font=\sffamily\small}]
    (1) edge[blue, line width=0.165cm] node [left] {} (2)
    (2) edge[blue, line width=0.165cm] node [left] {} (1)
    (4) edge[blue, line width=0.134cm] node [left] {} (5)
    	edge[red, line width=0.042cm]  node [left] {} (7)
    (5) edge[blue, line width=0.134cm] node [left] {} (4)
    (6) edge[blue, line width=0.119cm] node [left] {} (8)
    (7) edge[red, line width=0.042cm]  node [left] {} (4)
    	edge[blue, line width=0.086cm] node [left] {} (8)
    (8) edge[blue, line width=0.119cm] node [left] {} (6)
    	edge[blue, line width=0.086cm] node [left] {} (7);
\end{tikzpicture}
\centering
\caption{Model 2.}
\label{fig:model 2 graph}
\end{subfigure}
\caption{Graphical models illustrating substantial dependence structure of the GP-specific random effects $\ba$ in each model.  Edges between $\a_i$ and $\a_j$ are included if the 95\% credible interval of the marginal posterior distribution of the $(i,j)$-th entry of $\Sig_{\ba}^{-1}$ does not contain 0.  Blue edges represent positive dependence while red edges represent negative dependence.  The thickness of the edges is proportional to the strength of the dependence.}
\label{fig:random effect graph}
\end{figure}

Figure~\ref{fig:correlation graph} is also instrumental in explaining the suitability of the contraceptive products for a patient in terms of substitute goods, which in consumer theory is defined as products with similar functions that can be used in place of each other.  For conciseness, we only focus on some important relationships illustrated in the graphical model.  The propensity to discuss pill products $(y^\ast_1, y^\ast_2)$ are independent of each other given the hormonal IUD and the vaginal ring $(y^\ast_5, y^\ast_6)$ by the Markov property since all paths from $y^\ast_1$ to $y^\ast_2$ pass through $(y^\ast_5, y^\ast_6)$, reflecting the use of these non-pill contraceptives as pill alternatives dictated by particular clinical conditions.  The clique formed between $(y^\ast_5, y^\ast_7, y^\ast_8)$ suggests dependence in the propensity to discuss the hormonal IUD, the copper IUD and the condoms.  In fact, the posterior correlation between the propensity scores for both the IUD methods $(y^\ast_5, y^\ast_7)$ is around 0.52 on average (see Appendix~\ref{app:correlation}), suggesting a high tendency of these products to be discussed together.  This also reflects the fact that these IUD methods are substitutes.  Noticeably, the propensity to discuss the hormonal injection and the hormonal implant $(y^\ast_4, y^\ast_5)$ exhibit the highest level of association as indicated by our model, with a mean posterior correlation of 0.59.  This indicates the likelihood of these two prominent LARC products being included together in discussions, and it is consistent with them being close substitutes for each other for many patients.

Figure~\ref{fig:random effect graph} can be interpreted in the same way as Figure~\ref{fig:correlation graph}, regarding the substitutability of different products but in the context of ready-to-wear choices.  This is because the random effects in (\ref{eqn:model 1}) characterise the persistence of GPs in discussing a particular product after observing the patient's attributes.  There are clear differences in the graphical structure when comparing Figures~\ref{fig:model 1 graph} and \ref{fig:model 2 graph}.  The changes in the dependence structure of the GP random effects arise because some of the persistence in product choices can be explained by GP characteristics.  For example, the tendency of GPs to include both the hormonal injection and the copper IUD $(\a_3, \a_7)$ as ready-to-wear choices is due to their age (see significance of GP characteristics in Appendix~\ref{app:estimation result patient}).  The posterior structure also provides some confidence that the random effects specification is useful in capturing important GP characteristics that are not directly observed.  Three clusters of products with substantial dependence in ready-to-wear choices are identified from the model after accounting for the observed GP characteristics.  Particularly relevant is the dependence between the hormonal IUD and the implant $(\a_4, \a_5)$.  There is positive correlation between these two LARCs, indicating the tendency for GP attitudes (either positive or negative) to be aligned.  A second cluster includes both of the pills $(\a_1, \a_2)$ which is consistent with these products being used as a ready-to-wear default.  GPs who are more likely to discuss the combined pill after conditioning on the patient's attributes behave similarly when considering the mini-pill.  Contraceptives that are not pill- or hormone-based form the final bundle.

Our models allow us to examine posterior predictions for a range of patients.  Since we are interested in the uptake of LARC products, we specify a particular female patient where there is no clinical reason why a LARC should not be considered for discussion.  Table~\ref{table:attributes} of Appendix~\ref{app:attributes} gives the attributes of this base-case patient.  Figure~\ref{fig:marginal probability} summarises the estimate of the predictive probability of a GP discussing a particular product, where the range of predictions shown is generated for all GPs in the sample based on Model 2.  For this particular base-case patient, there is considerable agreement amongst all GPs in the sample that the combined pill (product 1) is one of the most suitable products to be discussed, but they have much more variable views on the other products.  Amongst the LARCs (products 3, 4, 5 and 7), the hormonal injection (product 3) and the implant (product 4) are the products which are the most likely to be discussed, with the variability across GPs perhaps simply reflecting a view that they are good substitutes to each other, which is in fact what we find in Figure~\ref{fig:correlation graph}.  GPs could indeed have consistent views about the need to discuss LARCs, as they do with the combined pill, but they are divided on which of the LARC products to discuss.  To explore this possibility, the final column in Figure~\ref{fig:marginal probability} shows the predicted probability of the GPs discussing at least one of these two products, that is $\P(y_3 + y_4 \geq 1)$.  The results suggest that the GPs will discuss either product 3 or 4 (or both) with similar probability to the combined pill.  While this joint probability does indicate a median that is similar to that of discussing the combined pill, the variability across GPs remains much larger than that associated with the combined pill.  This evidence is consistent with the hypothesised resistance amongst some GPs to even discuss LARCs, let alone recommend them.
\begin{figure}[t!]
\onehalfspacing
\centering
\includegraphics[trim= 2cm 1.5cm 2cm 3.5cm,clip,width=0.6\textwidth]{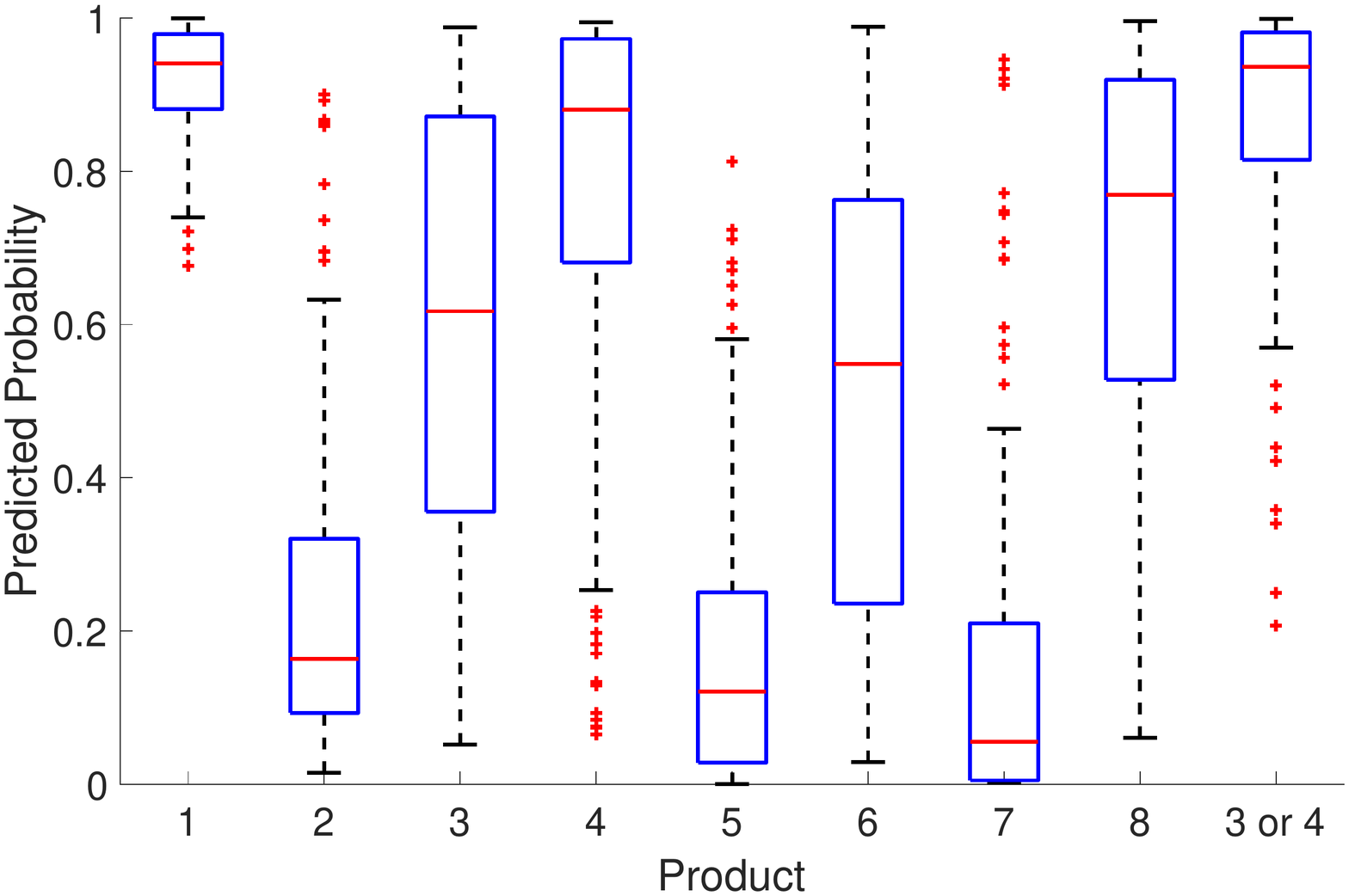}
\caption{Predicted probability of a GP discussing each product for a base-case patient for each of the 162 Australian GPs.}
\label{fig:marginal probability}
\end{figure}

\subsection{Comparing sampling schemes}

In order to investigate the performance of the antithetic sampler, Figure~\ref{fig:posterior application} illustrates marginal posterior distributions of those Model 2 parameters whose densities demonstrate the greatest visual differences between independent and antithetic sampling of the random effects $\ba_{1:P}$ and regression parameters $\bb$.  The marginal posterior distributions of $\a_{7,110}$ and $\b_{236}$ are effectively the same under both updating approaches.  This occurs because the mean of the conditional posterior distribution, which is a key ingredient in the deterministic antithetic sampler proposal, changes between iterations; a change largely driven by the stochastic update of the latent variable $\by^\ast$.  This outcome suggests that the posterior distribution of the other parameters remains adequately explored by the antithetic sampler.
\begin{figure}[t!]
    \centering
    \includegraphics[trim= 0.5cm 2.25cm 0.5cm 8cm,clip,width=0.8\textwidth]{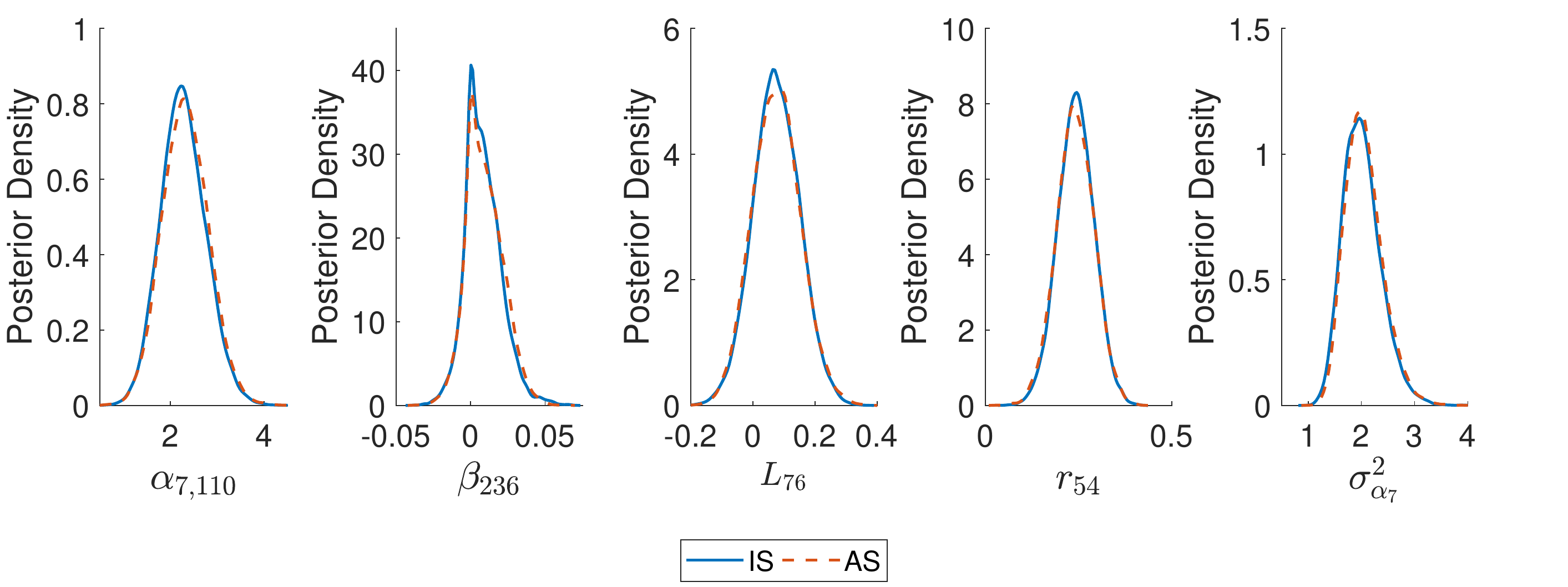}
    \caption{Marginal posterior density estimates of those Model 2 parameters with the greatest visual differences  between using independent sampling (IS) and antithetic sampling (AS) for $\ba_{1:P}$ and $\bb$.}
    \label{fig:posterior application}
\end{figure}

Table~\ref{table:IACT model 2} compares the performance between independent and antithetic sampling schemes when estimating Model 2.  The antithetic variable method generates samples marginally faster than independent sampling because it is deterministic.  Based on the results shown, we observe an improvement of 4.86 and 3.31 times performance gain on average in the mixing of $\ba_{1:P}$ and $\bb$ respectively.  As a result of this, the mean IACT of $\by^\ast$ is also improved.
\begin{table}[t!]
\centering
\doublespacing
\begin{tabular}{|c|r|r|r|r|r|} 
\hline
\multirow{2}{*}{Parameter} & \multicolumn{2}{c|}{Mean IACT} & \multicolumn{3}{c|}{IACT Ratio} \\ \cline{2-6}
 & \multicolumn{1}{c|}{IS} & \multicolumn{1}{c|}{AS} & \multicolumn{1}{c|}{Min} & \multicolumn{1}{c|}{Max} & \multicolumn{1}{c|}{Mean} \\ \hline
$\by^\ast$             & 3.6387  & 2.6686  & 0.8242 & 3.1419  & 1.2127 \\
$\ba_{1:P}$            & 16.8872 & 4.6456  & 1.4857 & 13.3424 & 4.8632 \\
$\bb$                  & 15.0446 & 4.0105  & 1.4566 & 16.0173 & 3.3111 \\
$\vechL(\Leps)$        & 14.8292 & 14.5422 & 0.9338 & 1.1737  & 1.0191 \\ 
$\vechL(\bR)$          & 12.7311 & 12.5170 & 0.9147 & 1.1509  & 1.0180 \\
$\diag(\Sig_{\ba})$    & 24.8056 & 14.6929 & 1.3130 & 2.0651  & 1.7222 \\ 
$\vechL(\bm{R}_{\ba})$ & 9.5025  & 5.1716  & 1.4599 & 2.3336  & 1.8424 \\ \hline
Time per iteration & 0.0243 & 0.0239 & \multicolumn{1}{c|}{-} & \multicolumn{1}{c|}{-} & \multicolumn{1}{c|}{-} \\ \hline
\end{tabular}
\caption{Comparison of the performance between independent sampling (IS) and antithetic sampling (AS) in the contraceptive products preference data in terms of the speed (seconds per iteration), the mean IACT and the IACT ratio for each block of parameter.}
\label{table:IACT model 2}
\end{table}

\section{Conclusion} \label{sec:conclusion}

Many methods exist for fitting a multinomial logit model with random effects, such as simulated maximum likelihood \citep{gong2004mobility}, quadrature \citep{hartzel2001multinomial,hedeker2003mixed}, multinomial-Poisson transformation \citep{lee2017poisson}, and moment-based estimation \citep{perry2017fast}, among others.  Computational strategies for the MVP model, on the other hand, are less well studied.  In this article, we introduce a HMC sampling approach to generate the posterior samples of $\bR$.  This method requires reparameterising $\bR$ into an unconstrained Cholesky factor in order to circumvent the restrictive properties of a correlation matrix having diagonal entries of 1 and being positive definite.  Furthermore, we propose a novel antithetic variable technique to accelerate the mixing of the random effects and the regression parameters, where significant gains in efficiency are observed in our application.  Although our antithetic sampling deterministically specifies the proposal distribution within the Metropolis-Hastings update, the ergodicity of the Markov chain is unaffected when it is embedded within a larger system of stochastic updates.

Our application considers the discussion of female contraceptive products by Australian GPs based on outcomes from the second stage of the stated preference data from \cite{fiebig2017consideration}.  An examination of the correlation matrix underlying the choices reveals a complex dependence structure between the products, hence indicating the plausibility of our formulation to model these choices in a multivariate setting.  Our empirical study also suggests evidence of medical practice variation among the GPs, especially with regard to the inclusion of LARCs in the discussion with patients.  The combined pill was the most popular contraceptive choice among the patients, and it represented a likely ready-to-wear default for many GPs.  Without GPs even discussing LARCs, their uptake was likely to remain relatively constrained in such a context.

\section*{Acknowledgements}

David Gunawan, Denzil Fiebig and Robert Kohn were partially supported by the Australian Research Council Discovery Project scheme DP150104630 and Scott Sisson was partially supported by the Discovery Project grant DP160102544.  Vincent Chin, David Gunawan, Robert Kohn and Scott Sisson were also partially supported by the Australian Centre of Excellence for Mathematical \& Statistical Frontiers (ACEMS) grant CE140100049.

\bibliographystyle{chicago}
\interlinepenalty=10000
\bibliography{biblio}

\newpage
\appendix

\section{Sampling scheme for the MVP model with random effects} \label{app:sampling scheme}

Suppose that we choose the following prior distributions: $\bb \sim \mathcal{N}({\bf 0}, \bm{\Psi}_{\bb}), \Sig_{\ba} \sim \mathcal{IW}(\l_{\Sig}, \bm{\Psi}_{\Sig})$ and the prior distribution on the lower triangular Cholesky factor $\Leps$ in (\ref{eqn:prior cholesky}) with $\nu=D+1$.  Let $\bt=(\bb, \Leps, \Sig_{\ba})$.  Equation (\ref{eqn:posterior}) gives the posterior distribution of interest under the data augmentation approach where we update $\by^\ast, \ba_{1:P}$ and each component of $\bt$ using Gibbs sampling.  For notational clarity, we will drop the superscript which indicates the sequence of the samples in a Markov chain where necessary.

\noindent
\newline
\underline{{\bf Step 1}: Updating $\by^\ast$}

\noindent
For $d=1,\dotsc,D$, sample $\by^\ast$ conditionally one-at-a-time following \cite{geweke1991efficient}, i.e.
\begin{equation*}
y^\ast_{d,it}|\ba_{1:P}, \bt, \by^\ast_{-d,it}, y_{d,it} \sim 
\begin{cases}
\mathcal{TN}_{(-\infty,0]}(\m^{(d|-d)}_{d,it}, \sigma^{(d|-d)}_{d,it}) & \text{if } y_{d,it}=0 \\
\mathcal{TN}_{(0,\infty)}(\m^{(d|-d)}_{d,it}, \sigma^{(d|-d)}_{d,it}) & \text{if } y_{d,it}=1
\end{cases}
\end{equation*}
where $\by^\ast_{-d,it}=(y_{1,it}, \dotsc, y_{d-1,it}, y_{d+1,it}, \dotsc, y_{D,it})^\top$, $\m^{(d|-d)}_{d,it}$ and $\sigma^{(d|-d)}_{d,it}$ are the univariate $d$-th dimension conditional mean and conditional standard deviation respectively for the $\N(\bmu_{it}, \bR)$ distribution and $\mathcal{TN}_{(a,b)}$ is a univariate normal distribution truncated to the interval $(a,b)$.

\noindent
\newline
\underline{{\bf Step 2}: Updating $\bb$}

\noindent
Compute the posterior mean $\bmu_{\bb}$ and the posterior covariance matrix $\Sig_{\bb}$ for $\bb$ as
\begin{equation*}
\begin{gathered}
\Sig_{\bb} = \Bigg(\sum_{i=1}^P \sum_{t=1}^T (\bm{I} \otimes \bx_{it}) \bR^{-1} (\bm{I} \otimes \bx_{it})^\top + \bm{\Psi}_{\bb}^{-1} \Bigg)^{-1}, \\
\bmu_{\bb} = \Sig_{\bb} \Bigg( \sum_{i=1}^P \sum_{t=1}^T (\bm{I} \otimes \bx_{it}) \bR^{-1} (\by^\ast_{it} - \ba_i) \Bigg),
\end{gathered}
\end{equation*}
where $\otimes$ denotes the Kronecker product and set $\bb^{[j+1]}=2\bmu_{\bb}-\bb^{[j]}$ deterministically.  If a horseshoe prior is specified on $\bb$ instead, its update is the same by first sampling $\diag(\bm{\Psi}_{\bb})$ conditional on the local shrinkage parameters $\l_i$ and global shrinkage parameter $\tau$ (see \cite{makalic2016simple} for details).

\noindent
\newline
\underline{{\bf Step 3}: Updating $\Leps$}

\noindent
Sample $\Leps$ using the NUTS algorithm and obtain the correlation matrix $\bR$ from the
relationship in (\ref{eqn:reparameterisation}).

\noindent
\newline
\underline{{\bf Step 4}: Updating $\ba_{1:P}$}

\noindent
For $i=1, \dotsc, P$, compute the posterior mean $\bmu_{\ba_i}$ and the posterior covariance matrix $\tilde{\Sig}_{\ba}$ for the random effects $\ba_i$ as
\begin{equation*}
\begin{gathered}
\tilde{\Sig}_{\ba} = \big(T\bR^{-1} + \Sig_{\ba}^{-1} \big)^{-1}, \\
\bmu_{\ba_i} = \tilde{\Sig}_{\ba} \Bigg(\bR^{-1} \sum_{t=1}^T \by^\ast_{it} - \bB \bx_{it} \Bigg),
\end{gathered}
\end{equation*}
and set $\ba^{[j+1]}_i = 2 \bmu_{\ba_i} - \ba^{[j]}_i$ deterministically.

\noindent
\newline
\underline{{\bf Step 5}: Updating $\Sig_{\ba}$}

\noindent
Sample
\begin{equation*}
\Sig_{\ba} \sim \mathcal{IW} \Bigg(\l_{\Sig}+P, \sum_{i=1}^P \ba_i \ba_i^\top + \Psi_{\Sig} \Bigg).
\end{equation*}
Suppose that a $\mathcal{HIW}(\l_{\Sig}, \bm{A})$ prior with scales $\bm{A}$ is used for $\Sig_{\ba}$. Sample
\begin{equation*}
\begin{gathered}
a_i \sim \mathcal{IG} \bigg(\frac{\l_{\Sig}+D}{2}, \l_{\Sig} \Sig_{\ba}^{-1}(i;i) + \frac{1}{A^2_i} \bigg), \\
\Sig_{\ba} \sim \mathcal{IW} \Bigg(\l_{\Sig}+P+D-1, \sum_{i=1}^P \ba_i \ba_i^\top + 2\l_{\Sig} \diag \bigg(\frac{1}{a_1}, \dotsc, \frac{1}{a_D} \bigg) \Bigg),
\end{gathered}
\end{equation*}
where $\Sig_{\ba}^{-1}(i;i)$ is the $i$-th diagonal entry of the precision matrix $\Sig_{\ba}^{-1}$.

\section{Attributes of the patient in the Australian GP data} \label{app:attributes}

\begin{table}[ht!]
\centering
\resizebox*{!}{\dimexpr\textheight-10\baselineskip\relax}{
\begin{tabular}{|c|c|l|}
\hline
Attribute & Variable & \multicolumn{1}{c|}{Description} \\ \hline
\multirow{4}{*}{Age} & dagegp1 & Aged 16-19 years \\ & \cellcolor{gray!25}dagegp2 & \cellcolor{gray!25}Aged 20-29 years \\
 & dagegp3 & Aged 30-39 years \\
 & dagegp4 & Aged 40 years or more \\
\hline
\multirow{4}{*}{Reason for encounter} & \cellcolor{gray!25}drfe1 & \cellcolor{gray!25}Starting prescribed contraception for first time \\
 & drfe2 & Recommencing prescribed contraception \\
 & drfe3 & On pill but dissatisfied \\
 & drfe4 & Using non-pill method but dissatisfied \\
\hline
\multirow{3}{*}{Periods} & dbleed1 & Heavy and/or painful periods \\
 & dbleed2 & Irregular periods \\
 & \cellcolor{gray!25}dbleed3 & \cellcolor{gray!25}No problems with periods \\
\hline
\multirow{3}{*}{Blood pressure} & dbp1 & Has low blood pressure \\
 & \cellcolor{gray!25}dbp2 & \cellcolor{gray!25}Has normal blood pressure \\
 & dbp3 & Elevated blood pressure \\
\hline
\multirow{4}{*}{Relationship} & drel1 & In long-standing relationship \\
 & \cellcolor{gray!25}drel2 & \cellcolor{gray!25}In new relationship \\
 & drel3 & Has no steady relationship \\
 & drel4 & No information about relationship \\
\hline
\multirow{3}{*}{Children} & dchild1 & Is currently breastfeeding \\
 & dchild2 & Has children but is not breastfeeding \\
 & \cellcolor{gray!25}dchild3 & \cellcolor{gray!25}Has no children \\
\hline
\multirow{4}{*}{Fertility plans} & dfut1 & Does not want to have children in future \\
 & dfut2 & Plans to have children in next 2 years \\
 & \cellcolor{gray!25}dfut3 & \cellcolor{gray!25}Plans to have children but not in next 2 years \\
 & dfut4 & Unsure about future fertility plans \\
\hline
\multirow{3}{*}{Pill preference} & dpil1 & Prefer pill to other methods \\
 & \cellcolor{gray!25}dpil2 & \cellcolor{gray!25}Has no strong opinion about pill \\
 & dpil3 & Prefers methods other than pill \\
\hline
\multirow{2}{*}{Weight concern} & dwt1 & Is concerned about gaining weight \\
 & \cellcolor{gray!25}dwt2 & \cellcolor{gray!25}Is not concerned about gaining weight \\
\hline
\multirow{2}{*}{Compliance} & \cellcolor{gray!25}dcomp1 & \cellcolor{gray!25}Has no difficulty with compliance \\
 & dcomp2 & Has difficulty with compliance \\
\hline
\multirow{3}{*}{Income} & \cellcolor{gray!25}dpay1 & \cellcolor{gray!25}Has a low to middle household income \\
 & dpay2 & Has a health care card \\
 & dpay3 & Has a high household income \\
\hline
\multirow{3}{*}{Smoking} & \cellcolor{gray!25}dsmk1 & \cellcolor{gray!25}Is a non-smoker \\
 & dsmk2 & Smokes less than 10 cigarettes per day \\
 & dsmk3 & Smokes 10 or more cigarettes per day \\
\hline
\end{tabular}
}
\caption{Categorical variables in the contraceptive discussion data with a text description for each level of attribute.  Levels in grey define the attributes of a base-case patient.}
\label{table:attributes}
\end{table}

\section{Posterior means of the patient and GP fixed effects in the Australian GP data based on Model 2} \label{app:estimation result patient}

\begin{table}[ht!]
\centering
\resizebox*{\textwidth}{!}{
\begin{tabular}{|c|c|rrrrrrrr|}
\hline
& \multirow{2}{*}{Variable} & \multicolumn{8}{c|}{Product} \\ \cline{3-10}
& & \multicolumn{1}{c}{1}   & \multicolumn{1}{c}{2}   & \multicolumn{1}{c}{3}   & \multicolumn{1}{c}{4}   & \multicolumn{1}{c}{5}   & \multicolumn{1}{c}{6}   & \multicolumn{1}{c}{7}   & \multicolumn{1}{c|}{8}   \\ \hline
\parbox[t]{2mm}{\multirow{27}{*}{\rotatebox[origin=c]{90}{Patient}}} & Intercept & \cellcolor{gray!25}1.4161  & \cellcolor{gray!25}-1.2576 & -0.3964 & \cellcolor{gray!25}1.0991  & \cellcolor{gray!25}-2.3943 & -0.1142 & \cellcolor{gray!25}-1.7657 & 0.6918  \\
& dagegp1   & \cellcolor{gray!25}0.1949  & -0.1329 & 0.0104  & 0.0744  & \cellcolor{gray!25}-0.5063 & -0.0205 & \cellcolor{gray!25}-0.2880 & 0.0637  \\
& dagegp3   & -0.1326 & 0.0621  & -0.0624 & -0.0002 & \cellcolor{gray!25}0.3173  & -0.0037 & 0.0906  & 0.0108  \\
& dagegp4   & \cellcolor{gray!25}-0.3936 & \cellcolor{gray!25}0.1851  & \cellcolor{gray!25}-0.2406 & -0.1041 & \cellcolor{gray!25}0.8095  & -0.0270 & \cellcolor{gray!25}0.3849  & 0.0013  \\
& drfe2     & -0.0426 & 0.0008  & -0.0388 & -0.0144 & 0.0441  & -0.0188 & -0.0449 & 0.0068  \\
& drfe3     & \cellcolor{gray!25}-0.2464 & -0.0541 & 0.0270  & 0.0788  & 0.0940  & 0.1069  & -0.0248 & \cellcolor{gray!25}0.1364  \\
& drfe4     & -0.0206 & 0.1042  & -0.0099 & 0.0516  & 0.0678  & 0.0719  & -0.0702 & 0.0056  \\
& dbleed1   & 0.0493  & \cellcolor{gray!25}-0.1363 & 0.0615  & -0.0869 & \cellcolor{gray!25}0.4000  & -0.0256 & \cellcolor{gray!25}-0.5274 & \cellcolor{gray!25}-0.2311 \\
& dbleed2   & 0.0160  & -0.0763 & 0.0213  & -0.0222 & 0.0070  & 0.0408  & -0.0869 & -0.0254 \\
& dbp1      & -0.0599 & -0.0011  & -0.0300 & 0.0292  & 0.0040  & 0.0317  & -0.0221 & \cellcolor{gray!25}-0.1433 \\
& dbp3      & \cellcolor{gray!25}-0.9956 & \cellcolor{gray!25}0.2444  & 0.0070  & 0.0135  & \cellcolor{gray!25}0.2375  & \cellcolor{gray!25}-0.2959 & \cellcolor{gray!25}0.2561  & 0.0347  \\
& drel1     & 0.0436  & -0.0102 & -0.0963 & -0.0020 & \cellcolor{gray!25}0.1570  & 0.0314  & 0.0282  & \cellcolor{gray!25}-0.3971 \\
& drel3     & -0.0141 & 0.0269  & -0.0208 & 0.0002  & -0.0271 & 0.0090  & -0.0186 & 0.0198  \\
& drel4     & -0.0914 & 0.0879  & 0.0667  & -0.0009 & -0.0101 & 0.0294  & 0.0029  & \cellcolor{gray!25}-0.2035 \\
& dchild1   & \cellcolor{gray!25}-1.7437 & \cellcolor{gray!25}1.3074  & -0.0082 & -0.0889 & \cellcolor{gray!25}0.9236  & \cellcolor{gray!25}-0.9909 & \cellcolor{gray!25}0.5354  & -0.0371 \\
& dchild2   & -0.0458 & 0.0344  & -0.0632 & -0.0403 & \cellcolor{gray!25}0.9850  & -0.0498 & \cellcolor{gray!25}0.6007  & -0.0543 \\
& dfut1     & \cellcolor{gray!25}-0.3206 & -0.0043 & \cellcolor{gray!25}0.1978  & 0.0245  & \cellcolor{gray!25}0.6323  & -0.0786 & \cellcolor{gray!25}0.2120  & -0.1143 \\
& dfut2     & \cellcolor{gray!25}-0.2861 & \cellcolor{gray!25}0.1936  & \cellcolor{gray!25}-0.2169 & \cellcolor{gray!25}-0.1996 & -0.0068 & 0.0359  & -0.1438 & 0.0116  \\
& dfut4     & \cellcolor{gray!25}-0.3591 & 0.0485  & 0.0470  & 0.0099  & \cellcolor{gray!25}0.2882  & 0.0067  & 0.0150  & 0.0323  \\
& dpil1     & \cellcolor{gray!25}0.4724  & \cellcolor{gray!25}0.3662  & -0.0948 & \cellcolor{gray!25}-0.2629 & -0.0120 & -0.0331 & -0.0430 & -0.0287 \\
& dpil3     & \cellcolor{gray!25}-0.1878 & \cellcolor{gray!25}-0.2417 & 0.0289  & 0.0618  & 0.0538  & 0.0329  & 0.0457  & 0.0814  \\
& dwt1      & 0.0831  & 0.0374  & \cellcolor{gray!25}-0.2582 & -0.0624 & 0.0318  & 0.0652  & -0.0130 & 0.0815  \\
& dcomp2    & \cellcolor{gray!25}-0.3401 & \cellcolor{gray!25}-0.1988 & \cellcolor{gray!25}0.2152  & 0.0642  & \cellcolor{gray!25}0.2321  & -0.0033 & \cellcolor{gray!25}0.3133  & -0.0162 \\
& dpay2     & -0.0253 & -0.0558 & -0.0204 & -0.0026 & 0.0084  & 0.0595  & 0.0082  & 0.0074  \\
& dpay3     & 0.0317  & -0.0639 & -0.0697 & -0.0177 & -0.0373 & \cellcolor{gray!25}0.2896  & -0.0177 & -0.0044 \\
& dsmk2     & \cellcolor{gray!25}-0.2665 & -0.0117 & -0.0266 & -0.0126 & -0.0038 & 0.0444  & 0.0892  & 0.0320  \\
& dsmk3     & \cellcolor{gray!25}-0.5218 & -0.0133 & 0.0132  & 0.0255  & 0.0148  & -0.0546 & 0.0467  & 0.0333  \\ \hline
\parbox[t]{2mm}{\multirow{7}{*}{\rotatebox[origin=c]{90}{GP}}} & Female & -0.0662 & 0.0248 & \cellcolor{gray!25}-0.4417 & 0.0732 & 0.0368 & \cellcolor{gray!25}0.5999 & -0.4474 & -0.0260 \\
& Fellow & -0.0183 & -0.0958 & 0.0709 & 0.0418 & 0.2067 & 0.1019 & -0.1456 & -0.0108 \\
& Family planning & -0.0002 & -0.0154 & -0.1203 & 0.2229 & 0.0434 & 0.0360 & -0.0324 & -0.0118 \\
& Bulk-bill & -0.0210 & -0.0349 & 0.0416 & -0.0372 & -0.0617 &  0.0036 & 0.0509 & 0.0038 \\
& Age & 0.0086 & 0.0080 & \cellcolor{gray!25}0.0207 & -0.0061 & \cellcolor{gray!25}0.0175 & -0.0044 & 0.0093 & -0.0100 \\
& Australian graduate & 0.0839 & 0.0564 & -0.0087 & 0.3466 & 0.0911 & -0.2385 & -0.0965 & \cellcolor{gray!25}0.5515 \\
& Urban & -0.0888 & 0.0065 & 0.0706 & -0.0078 & 0.0099 & 0.0048 & -0.0222 & 0.1774 \\ \hline
\end{tabular}
}
\caption{Regression coefficient posterior mean estimates for the attributes of a female patient and the characteristics of a GP based on Model 2 for various products in the contraceptive discussion data.  Parameters whose 90\% credible interval does not include 0 are shown in grey.}
\label{table:estimation result for women attributes}
\end{table}

\section{Posterior mean of \texorpdfstring{$\bR$}{bReps} in the Australian GP data based on Model 2} \label{app:correlation}
\begin{equation*}
    \bR = \left[
    \begin{tabular}{cccccccc}
        { 1.0000} & -0.1126   & -0.0515   & -0.0450   & -0.2349   & { 0.4712} & -0.2065   & -0.0204   \\
        -0.1126   & { 1.0000} & { 0.1625} & { 0.0449} & -0.0263   & -0.2679   & -0.0537   & -0.0494   \\
        -0.0515   & { 0.1625} & { 1.0000} & { 0.5873} & { 0.1779} & { 0.0153} & { 0.1836} & { 0.0189} \\
        -0.0450   & { 0.0449} & { 0.5873} & { 1.0000} & { 0.2414} & { 0.0379} & { 0.1889} & { 0.1048} \\
        -0.2349   & -0.0263   & { 0.1779} & { 0.2414} & { 1.0000} & -0.0696   & { 0.5177} & -0.0771   \\
        { 0.4712} & -0.2679   & { 0.0153} & { 0.0379} & -0.0696   & { 1.0000} & -0.0055   & { 0.1831} \\
        -0.2065   & -0.0537   & { 0.1836} & { 0.1889} & { 0.5177} & -0.0055   & { 1.0000} & { 0.2058} \\
        -0.0204   & -0.0494   & { 0.0189} & { 0.1048} & -0.0771   & { 0.1831} & { 0.2058} & { 1.0000} \\
    \end{tabular}
    \right]
\end{equation*}

\section{Posterior mean of \texorpdfstring{$\Sig_{\ba}$}{Sigba} in the Australian GP data based on Model 2} \label{app:covariance}
\begin{equation*}
    \Sig_{\ba} = \left[
    \begin{tabular}{cccccccc}
        0.5574 & 0.3005 & { 0.2760} & { 0.2490} & { 0.0795} & 0.2056 & { 0.0634} & 0.2592 \\
        0.3005 & 0.6923 & { 0.3040} & { 0.2679} & { 0.1875} & 0.2199 & { 0.2418} & 0.3358 \\
        0.2760 & 0.3040 & { 1.3574} & { 0.2590} & -0.0188   & 0.1065 & { 0.2586} & 0.0751 \\
        0.2490 & 0.2679 & { 0.2590} & { 1.6084} & { 0.5244} & 0.2538 & -0.2229   & 0.2383 \\
        0.0795 & 0.1875 & -0.0188   & { 0.5244} & { 1.1040} & 0.2911 & { 0.0135} & 0.2612 \\
        0.2056 & 0.2199 & { 0.1065} & { 0.2538} & { 0.2911} & 1.5142 & { 0.2950} & 0.4906 \\
        0.0634 & 0.2418 & { 0.2586} & -0.2229   & { 0.0135} & 0.2950 & { 2.0530} & 0.4144 \\
        0.2592 & 0.3358 & { 0.0751} & { 0.2383} & { 0.2612} & 0.4906 & { 0.4144} & 1.2942 \\
    \end{tabular}
    \right]
\end{equation*}
\end{document}